\documentclass[12pt,tightenlines,eqsecnum,floats,showpacs,nofootinbib, aps, prd]{revtex4-2}
\usepackage{enumitem}
\usepackage{graphicx}
\usepackage{subcaption}
\usepackage{amsmath,amsthm}
\usepackage{amsfonts}
\usepackage{amssymb}
\usepackage{mathrsfs}
\usepackage[dvipsnames]{xcolor}
\usepackage{color}
\usepackage{epstopdf}
\usepackage[utf8]{inputenc}
\usepackage{natbib}
\usepackage[colorlinks,urlcolor=NavyBlue,citecolor=NavyBlue,linkcolor=NavyBlue,pdfusetitle]{hyperref}
\usepackage[all]{hypcap}
\usepackage{orcidlink}

\usepackage{wasysym}
\newtheorem{theorem}{Theorem}
\newtheorem*{theoremgb}{Theorem}
\newtheorem{lemma}{Lemma}
\newtheorem{proposition}{Proposition}

\newtheorem{definition}{Definition}

\theoremstyle{remark}
\newtheorem*{remark}{Remark}
\newtheorem{example}{Example}

\newcommand{\scri}{\mathscr{I}}
\newcommand{\scrip}{\scri^{+}}
\newcommand{\scrim}{\scri^{-}}

\renewcommand{\th}{\textrm{\upshape{\thorn}}}

\usepackage{CJKutf8}
\newcommand{\Ro}{\textrm{\begin{CJK}{UTF8}{min}ロ\end{CJK}}}

\begin{document}

\title{Newman-Penrose-like exact and approximate conservation laws: a covariant and conformal formulation}
\author{Berend Schneider\orcidlink{0009-0004-4099-2574}}
\email{berend@uoguelph.ca}
\affiliation{Department of Physics, University of Guelph, Guelph, Ontario, Canada N1G 2W1}

\begin{abstract}
Using a conformal extension of the Geroch-Held-Penrose (GHP) formalism I derive a manifestly covariant and conformal expression of Newman-Penrose (NP) constants, which are a set of conserved quantities associated to solutions to the wave equation on light cones in Minkowski space. The resulting expression generalizes to massless fields of arbitrary spin---including the electromagnetic field, Weyl fermions, and the linearized Weyl tensor---on conformally flat space-times. 

In some non-conformally flat space-times there may exist conserved charges on very special null hypersurfaces. Using the conformal GHP formalism I prove the existence of conserved Aretakis charges on extremal Killing horizons. 

In the absence of exact conservation laws it is still useful to extend the definition for NP constants to some asymptotically flat curved space-times, where the conservation laws become approximate conservation laws which rapidly approach exact conservation laws at null infinity. I derive explicit expressions for spherically symmetric space-times. 
\end{abstract}

\maketitle

\tableofcontents

\section{Introduction and summary}

Associated with solutions to the standard wave equation in Minkowski space, there are an infinite number of conservation laws on each light cone. The corresponding conserved quantities are known as \textit{Newman-Penrose (NP) constants}. In this paper, I derive a covariant and conformal representation of NP constants for massless fields of arbitrary spin---including the electromagnetic field, Weyl fermions, and the linearized Weyl tensor---using a novel conformal extension of the Geroch-Held-Penrose (GHP) formalism \cite{Geroch1973, Penrose1985, Penrose1986} (cf. \S5.6 of \cite{Penrose1985}, and \cite{Ludwig1988} for other proposed conformal extensions). This covariant representation may be extrapolated to curved space-times, where they define generalized NP charges satisfying \textit{approximate conservation laws}. \\
\\
The relevance of the approximate NP conservation laws is found in their application to the study of asymptotics. They can be used to propagate solutions away from an initial data surface. Physically, the `non-conservation' is due to the backscattering of radiation in the presence of curvature. To illustrate this, I will begin by sketching a derivation, originally due to Kehrberger \cite{Kehrberger2021a}, of the asymptotic form of a massless scalar field on a Schwarzschild background. The crux of the proof is to integrate an approximate conservation law, like the ones derived in this paper. \\
\\
On space-times that are not conformally flat there generally are no conservation laws, although there are a few interesting exceptions. On space-times that are asymptotically flat at null infinity $\scri$, conservation laws exist on $\scri$ for fields satisfying the massless free-field equations \cite{Newman1965,Newman1968,Exton1969}. Curiously, this includes the full non-linear Einstein vacuum equations. In contrast with the linearized Einstein equations, the conserved charges associated to the full Einstein equations are generally non-vanishing even in stationary space-times and thus have non-trivial physical content. A second place where conservation laws have been found is on the horizon $\mathcal{H}$ of extremal Reissner-Nordstr\"om black holes \cite{Aretakis2011}. Since $\scri$ in asymptotically flat space-times and $\mathcal{H}$ in the extremal Reissner-Nordstr\"om space-times are both conformally flat null hypersurfaces it is perhaps not entirely surprising that they contain conservation laws given the conformal invariance of the NP constants. What is more surprising is that conservation laws associated to solutions to the scalar wave equation were shown by Aretakis to exist on \textit{any} extremal Killing horizon (including the event horizon in the extremal Kerr space-time), even though this horizon is not conformally flat \cite{Aretakis2012}. Proving a similar result for fields of non-zero spin seems to be more complicated. Aretakis-like conserved electromagnetic- and gravitational charges were shown to exist on the extremal Kerr horizon by Lucietti and Reall \cite{Lucietti2012}, and on a few other examples of higher dimensional extremal black holes by Murata \cite{Murata2013}. Aretakis has further shown that the existence of conserved charges associated to solutions to the scalar wave equation is equivalent to the existence of solutions to a certain elliptic differential equation on the hypersurface \cite{Aretakis2013, Aretakis2014}. I extend this result to wave equations of arbitrary spin, and prove the existence of conserved charges associated to the scalar wave equation on extremal Killing horizons using the conformal GHP formalism. \\
\\
I follow the conventions of Penrose \& Rinder \cite{Penrose1985,Penrose1986}: Space-time is a four-dimensional Lorentzian manifold with metric $g_{ab}$ that has signature $(+,-,-,-)$. Tensor fields are denoted with lowercase Latin abstract indices $a, b, c, \dots$. 

\subsection{Newman-Penrose constants in Minkowski space}

Let us start with the simple case of the scalar wave equation $\Box\phi = 0$, given by
\begin{align}
    \partial_u\partial_v(r\phi_\ell) &= -\frac{\ell(\ell+1)}{r^2}(r\phi_\ell) ,
\end{align}
where $u := \tfrac{1}{2}(t - r)$ and $v := \tfrac{1}{2}(t + r)$ are the standard outgoing and ingoing null coordinates, and where $\phi_\ell$ is the $\ell$-mode of $\phi$ in the spherical harmonic decomposition $\phi = \sum_\ell\phi_\ell(u,v)Y_\ell$. Note that I am suppressing the azimuthal mode number $m$. Spherically symmetric ($\ell=0$) modes satisfy
\begin{align}
    \label{eq:flat_spherical_mode}
    \partial_u\partial_v(r\phi_{\ell=0}) &= 0 .
\end{align}
The solutions to Eq.~\eqref{eq:flat_spherical_mode} are given by a superposition of outgoing and ingoing waves; $\phi_{\ell=0} = r^{-1}U(u) + r^{-1}V(v)$. Purely outgoing waves satisfy $\partial_v(r\phi_{\ell=0}) = 0$, which motivates us to call $\partial_v(r\phi_{\ell=0})$ the ingoing radiation field. The wave equation Eq.~\eqref{eq:flat_spherical_mode} takes the form of a conservation law for ingoing radiation: 
\begin{align}
    \partial_v(r\phi_{\ell=0}) = \partial_v(r\phi_{\ell=0})|_{u=-\infty} ,
\end{align}
so that all ingoing waves can be traced back to $\scrim$. The wave equation for arbitrary $\ell$-modes does not directly take the form of a conservation law, but for each mode it is possible to construct the following conservation law: 
\begin{align}
    \label{eq:flat_NP_laws}
    \partial_u\Bigl(r^{-2\ell}\partial_v\bigl[(r^2\partial_v)^\ell(r\phi_\ell)\bigr]\Bigr) = 0 .
\end{align}
The conserved quantities $Q_\ell = r^{-2\ell}\partial_v\bigl[(r^2\partial_v)^\ell(r\phi_\ell)\bigr]$ are called \textit{Newman-Penrose (NP) constants}. 

\subsection{Newman-Penrose-like approximate conservation laws in Schwarzschild and their application}

To see how the picture changes in curved space-times, let us first consider spherically symmetric solutions to the scalar wave equation in the Schwarzschild space-time. In double null coordinates given by $ds^2 = 4fdudv - r^2dS^2$, where $f := 1 - 2mr^{-1}$ and $dS^2$ is the unit sphere metric, the wave equation for spherically symmetric solutions is given by:
\begin{align}
    \label{eq:Schw_spherical_mode}
    \partial_u\partial_v(r\phi) &= -\frac{2mf}{r^3}(r\phi) .
\end{align}
This now has the form of an approximate conservation law, which becomes the exact conservation law Eq.~\eqref{eq:flat_spherical_mode} in the Minkowski space limit $m \to 0$. The non-conservation may be understood physically as a consequence of the failure of the \textit{strong Huygens'} principle in the presence of curvature. Waves may backscatter off the curvature, so that initially purely outgoing radiation may evolve to acquire an ingoing component. \\
\\
I will now briefly sketch an argument, due to Kehrberger \cite{Kehrberger2021a}, connecting the asymptotics of initial data to the asymptotics of $\phi$ at $\scrip$, by integrating the approximate conservation law \eqref{eq:Schw_spherical_mode}: Suppose data is given on an ingoing null hypersurface $v = 0$ by $r\phi|_{v=0,u=-\infty} = 0$ and $\partial_u(r\phi)|_{v=0} = O(|u|^{-p})$ where $p > 1$ is a constant, and on $\scrim$ by the no ingoing radiation condition $\partial_v(r\phi)|_{u=-\infty} = 0$. Suppose $\phi$ satisfies the preliminary global estimate $|\phi| \lesssim r^{-\frac{1}{2}}$, where the notation $f \lesssim g$ means that there exists a sufficiently large constant $C$ such that $f \leq Cg$. We may improve this estimate using the `approximate conservation law' given by Eq.~\eqref{eq:Schw_spherical_mode}. We may simply integrate Eq.~\eqref{eq:Schw_spherical_mode} outwards from $v=0$ to obtain: 
\begin{align}
    \label{eq:estimate_1}
    |\partial_u(r\phi) - \partial_u(r\phi)|_{v=0}| \leq \int^v_0|\partial_v\partial_u(r\phi)|dv = \int^v_0|2mfr^{-2}\phi|dv \lesssim |u|^{-\frac{3}{2}} ,
\end{align}
so that $|\partial_u(r\phi)| \lesssim |u|^{-p} + |u|^{-\frac{3}{2}}$, and then integrate once more, inwards from $\scrim$: 
\begin{align}
    \label{eq:estimate_2}
    |r\phi - r\phi|_{u=-\infty}| \leq \int^u_{-\infty}|\partial_u(r\phi)|du \lesssim |u|^{1-p} + |u|^{-\frac{1}{2}} ,
\end{align}
so that $|r\phi| \lesssim |u|^{1-p} + |u|^{-\frac{1}{2}}$. If $p \leq \frac{3}{2}$, $|r\phi| \lesssim |u|^{1-p}$. Otherwise, iteratively inserting the improved estimate back into Eq.~\eqref{eq:estimate_1} $\lceil p-\frac{3}{2} \rceil$ times yields the estimate $|\partial_u(r\phi)| \lesssim |u|^{-p}$. With this estimate for outgoing radiation, we may derive the asymptotic form of ingoing radiation due to backscattering: 
\begin{align}
    |\partial_v(r\phi)| &\leq \int^u_{-\infty}|2mfr^{-2}\phi|du \lesssim |u|^{-1-p} .
\end{align}
The point I want to demonstrate here is that even though Eq.~\eqref{eq:Schw_spherical_mode} is not a conservation law, we were still able to find the asymptotic behaviour of $\phi$ because it is very close to being a conservation law at large $r$. If the right-hand side of Eq.~\eqref{eq:Schw_spherical_mode} decayed merely as $r^{-2}$ instead of as $r^{-3}$, then Eq.~\eqref{eq:estimate_1} would have yielded $|\partial_u(r\phi)| \lesssim |u|^{-\frac{1}{2}}$, so that $|\phi| \lesssim |u|^{-\frac{1}{2}}$ by Eq.~\eqref{eq:estimate_2}, which is not an improvement over the initial estimate $|\phi| \lesssim r^{-\frac{1}{2}}$. This is precisely the problem with the higher modes, for which the wave equation becomes
\begin{align}
    \partial_u\partial_v(r\phi) &= -\frac{\ell(\ell+1)f}{r^2}(r\phi) - \frac{2mf}{r^3}(r\phi) .
\end{align}
The problematic $r^{-2}$ term is also present in the flat wave equation. There, the solution is to instead integrate the \textit{NP conservation laws} \eqref{eq:flat_NP_laws}. On the Schwarzschild space-time, these conservation laws don't exist but nevertheless approximate conservation laws exist similar to Eq.~\eqref{eq:Schw_spherical_mode} for higher modes---see Eq. (7.10) of \cite{Kehrberger2021c}. These `generalized NP conservation laws' are much less compact than the flat NP conservation laws \eqref{eq:flat_NP_laws}, and the complexity escalates rather quickly when considering Kerr or more complicated space-times---see \S4 of \cite{Angelopoulos2021}. One of the purposes of this paper is to provide an algorithm for computing NP-like approximate conservation laws, like Eq.~\eqref{eq:Schw_spherical_mode} for fields of arbitrary spin in spherically symmetric space-times, including the Schwarzschild space-time, using the conformal GHP formalism. They generalize the approximately conserved quantities derived by Kehrberger for the scalar wave equation in \cite{Kehrberger2021c} and the Regge-Wheeler equation in \cite{Kehrberger2024a}. 

\subsection{Summary of the main results}
\subsubsection{Conformal GHP operators}
Here I define conformal GHP operators which I will need to state the main results. A significantly more detailed exposition will be given in \S\ref{sec:conformal_GHP}. For a short overview of the GHP formalism refer to Appendix \ref{appendix:GHP}. 

\begin{definition}
A function $\eta$ is said to be a conformal- and GHP-weighted function with weights $\{w, p, q\}$ if, under tetrad transformations of the form
\begin{align}
    (\hat{l}_a, \hat{n}_a, \hat{m}_a) &= (\Omega\lambda\bar{\lambda}l_a, \Omega\lambda^{-1}\bar{\lambda}^{-1}n_a, \Omega\lambda\bar{\lambda}^{-1}m_a) , \\
    \hat{\eta} &= \Omega^w\lambda^p\bar{\lambda}^q\eta .
\end{align}
\end{definition}
\noindent The spin coefficients $\kappa$, $\sigma$, $\rho - \bar{\rho}$, $\tau - \bar{\tau}'$, and their primed variants are all conformal with conformal weight $w=-1$. The remaining spin coefficients can be combined with the GHP operators to form conformal weighted GHP operators
\begin{subequations}
\begin{align}
    \th_c &:= \th + [w - \tfrac{1}{2}(p+q)]\rho - q(\rho - \bar{\rho}) \\
    \eth_c &:= \eth + [w - \tfrac{1}{2}(p-q)]\tau + q(\tau - \bar{\tau}') . 
\end{align}
\end{subequations}
$\th_c$ and $\eth_c$ both have conformal weight $w=-1$. \\
\\
The generalized wave equation I will be considering is
\begin{align}
    \label{eq:generalized_wave_eq}
    \Ro\phi_k &= 0 ,
\end{align}
where $\Ro := \th_c\th_c' - \eth_c\eth_c'$ is a weighted wave operator generalizing the standard wave operator $\Box$. When acting on quantities with weights $\{-1, 0, 0\}$, $2(\Ro - \Psi_2) = \Box + \tfrac{1}{6}R$. Eq.~\eqref{eq:generalized_wave_eq} contains many physically interesting wave equations, including the source-free Maxwell equations and the linearized vacuum Einstein equations. 

\subsubsection{The main theorems}
I am now in a position to state the main results of this paper. A covariant and conformal representation of the NP conservation laws \eqref{eq:flat_NP_laws} in Minkowski space is given by: 
\begin{theorem}
    Let $\phi_k$ be a scalar with weights $\{-s-1, 2(s-k), 0\}$ satisfying $\Ro\phi_k = 0$. Let $\mathcal{N}$ be a light cone in Minkowski space generated by $n^a$. Choose a foliation of $\mathcal{N}$ by spherical cuts $\mathcal{S}$ so that $\mathcal{N}$ and $\mathcal{S}$ are spherically symmetric with respect to a common centre. If $\mathcal{S}$ and $\mathcal{S}'$ are any two cuts belonging to this foliation, then
    \begin{align}
        \label{eq:cov_and_conf_expr}
        \oint_\mathcal{S}Y\th_c^{\ell-s+k+1}\phi_k\mathcal{S} &= \oint_{\mathcal{S}'}Y\th_c^{\ell-s+k+1}\phi_k\mathcal{S} , 
    \end{align}
    where $\mathcal{S}$ in the integrand denote the area forms on $\mathcal{S}$ and $\mathcal{S}'$, and where $Y = \eth_c^{\ell-s+k}Z$ is a weighted function with weights
    \begin{align}
        Y &: \{\ell+k,-\ell-s+k-1,-\ell+s-k-1\} . 
    \end{align}
    $Z$ satisfies $\th_c'Z = 0 = \eth_c'Z$. In terms of spin-weighted spherical harmonics, $Y$ consists of ${}_{k-s}Y_\ell$. 
\end{theorem}

\noindent Since the expression \eqref{eq:cov_and_conf_expr} is conformal, this conservation law naturally also holds in any conformally flat space-time. \\
\\
The next result gives a necessary and sufficient condition for the existence of conservation laws of the form \eqref{eq:cov_and_conf_expr} with $\ell = s-k$, for \textit{some} function $Y$, on an arbitrary curved background.

\begin{theorem}
Associated to solutions $\phi_k$ to the spin-$s$ wave equation $(\Ro - V)\phi_k = 0$, where $V$ is an arbitrary function, there exist conserved quantities on a null hypersurface $\mathcal{N}$ generated by $l^a$ of the form
\begin{align}
    \label{eq:Aretakis_like_charge}
    \oint_\mathcal{S}\omega\th_c'\phi_k\mathcal{S}
\end{align}
if and only if there exists a function $\omega$ with weights $\{s, 2(k-s)+1, 1\}$ satisfying $\th_c\omega = 0$ and the elliptic differential equation
\begin{align}
\label{eq:elliptic_pde}
    \bigl[\eth_c'(\eth_c - (\tau - \bar{\tau}')) + V\bigr]\omega &= 0 .
\end{align}
\end{theorem}

\noindent A similar result was first obtained by Aretakis in the special case of the scalar wave equation $\Box\phi=0$, i.e. in the case $s=k=0$ and $V = \Psi_2 + 
2\Lambda$ \cite{Aretakis2013,Aretakis2014}. In this case, Aretakis also proved that solutions to equation \eqref{eq:elliptic_pde} exist on extremal Killing horizons. I will state and prove his result using the conformal GHP formalism. 

\begin{theorem}[Aretakis \cite{Aretakis2013}]
\label{thm:Aretakis2013}
Let $\phi$ be a solution to the wave equation $\Box\phi = 0$, and let $\mathcal{H}$ be an extremal Killing horizon---a null hypersurface generated by a Killing vector field $k^a \overset{\mathcal{H}}{=} kl^a$ satisfying the extremality condition $k^b\nabla_bk^a \overset{\mathcal{H}}{=} 0$. Let $\mathcal{H}$ be foliated by closed space-like cross-sections $\mathcal{S}$ whose intrinsic metric $2m_{(a}\bar{m}_{b)}$ satisfies $\mathcal{L}_k(2m_{(a}\bar{m}_{b)}) \overset{\mathcal{H}}{=} 0$. Then there exists a positive function $\omega$ on $\mathcal{H}$, which is unique up to a factor, such that
\begin{align}
    \oint_\mathcal{S}\omega\th_c'\phi\mathcal{S}
\end{align}
is independent of the cross section $\mathcal{S}$. 
\end{theorem}

\noindent Even in the absence of conservation laws, Newman-Penrose-like approximate conservation laws can still be a useful tool to obtain the asymptotic structure of solutions to wave equations à la Kehrberger. In this paper I only deal with the spherically symmetric case in detail although I hope the methods may be aplied to more general space-times. 

\begin{theorem}
\label{thm:spherical_NP_flux_balance}
Let $\phi_k$ satisfy $(\Ro - V)\phi_k = 0$. On non-flat spherically symmetric null hypersurfaces $\mathcal{N}$, the NP conservation laws \eqref{eq:cov_and_conf_expr} become the flux-balance laws
\begin{multline}
    \label{eq:NP_flux_balance_spherical}
    \oint_{\mathcal{S}'}Y\th_c^{\ell-s+k+1}\phi_k\mathcal{S} - \oint_\mathcal{S}Y\th_c^{\ell-s+k+1}\phi_k\mathcal{S} \\
    = \int_\Sigma Y\Bigl(\th_c^{\ell-s+k}(V\phi_k) + \sum_{n=0}^{\ell-s+k}6(n+s-k)\th_c^{\ell-s+k-n}\bigl[\Psi_2\th_c^n\phi_k\bigr]\Bigr)\mathcal{N} ,
\end{multline}
where $\Sigma \subset \mathcal{N}$ has future boundary $\mathcal{S}'$ and past boundary $\mathcal{S}$. 
\end{theorem}

\subsection{Future directions}
\subsubsection{Aretakis-like charges for higher spin fields on extremal horizons}
It is natural to ask if theorem \ref{thm:Aretakis2013} can be generalized to fields of higher spin. The proof of theorem \ref{thm:Aretakis2013} uses the fact that in the case $s=k=0$, the elliptic operator $\mathcal{O}$ defined by \eqref{eq:elliptic_pde} is real so that it has a unique (up to a factor) positive principal eigenfunction. In general, however, $\mathcal{O}$ is complex. Another interesting case for which $\mathcal{O}$ is real is when $k = s$, but even in this case I was not able to prove that solutions to equation \eqref{eq:elliptic_pde} exist on extremal horizons when $s\neq0$ for any interesting choices of $V$, and I suspect that generally there are none. There is evidence that suggests that conserved charges of the form \eqref{eq:Aretakis_like_charge} exist when $k = 0$. The existence of conserved electromagnetic and gravitational charges was shown by Lucietti and Reall \cite{Lucietti2012} on extremal Kerr horizons, and by Murata \cite{Murata2013} on some other examples of extremal horizons in higher dimensions. In both of these works, conserved charges were found of the form \eqref{eq:Aretakis_like_charge} with $k=0$. Furthermore, Lucietti and Reall found conserved charges of the form
\begin{align}
    \oint\omega\th_c'^{2s+1}\phi_{2s}\mathcal{S} ,
\end{align}
for integer $s$ (cf. Eq.~\eqref{eq:cov_and_conf_expr} with $k=2s$, $\ell = k - s$), and also higher order conserved charges of the form
\begin{align}
    \oint\omega\th_c'^{\ell\pm s+1}\phi_{s\pm s}\mathcal{S} ,
\end{align}
for integer $s$ (cf. Eq.~\eqref{eq:cov_and_conf_expr} with $k=0,2s$). Earlier, higher order conserved charges associated to solutions to the scalar wave equation on extremal horizons in spherical symmetry were also found by Aretakis \cite{Aretakis2012}. It would be nice to obtain the conditions for which conservation laws exist for fields of arbitrary spin, and to find a unifying expression for the conserved charges including the ones found by Aretakis, Lucietti and Reall. 

\subsubsection{Newman-Penrose-like approximate conservation laws in general asymptotically flat space-times}
Spherically symmetric space-times are some of the simplest curved space-times. The ultimate goal is to obtain NP-like approximate conservation laws in more complicated space-times like Kerr. The Kerr space-time is Petrov type D, like spherically symmetric space-times. Unlike in spherically symmetry, however, the principal null directions in Kerr do not generate null hypersurfaces. A problem that arises on shearing null hypersurfaces is that there are no functions $Z$ that simultaneously satisfy $\th_c'Z = 0$ and $\eth_c'Z = 0$ (cf. Eq.~\eqref{eq:integrability_condition}). This means that the function $Y$ as defined in theorem \ref{thm:covariant_NP_conservation_law} is not a suitable definition in general curved space-times. \\
\\
Nevertheless, there are some indications that suggest that a conformal approach may work for general asymptotically flat space-times. In the Kerr space-time, the conformal spin coefficients fall off faster than their non-conformal counterparts. For example, in Boyer-Lindquist coordinates, in the Kinnersley tetrad, they are given by
\begin{subequations}
\begin{align}
    \rho &= -\frac{1}{r-ia\cos\theta} & \implies \rho &= O(r^{-1}) , & \rho - \bar{\rho} &= O(r^{-2}) , \\
    \rho' &= \frac{r^2-2Mr+a^2}{2(r+ia\cos\theta)(r-ia\cos\theta)^2} & \implies \rho' &= O(r^{-1}) , & \rho' - \bar{\rho}' &= O(r^{-2}) , \\
    \tau &= -\frac{ia\sin\theta}{\sqrt{2}(r^2+a^2\cos^2\theta)} & \implies \tau &= O(r^{-2}) , \\
    \tau' &= -\frac{ia\sin\theta}{\sqrt{2}(r-ia\cos\theta)^2} & \implies \tau' &= O(r^{-2}) , & \tau - \bar{\tau}' &= O(r^{-3}) . 
\end{align}
\end{subequations}
In each case, the conformal spin coefficients fall off by an extra power of $r$. The fall-off rate is the same in Bondi coordinates. In fact, in Bondi coordinates, in the Newman-Unti tetrad, the fall-off rates of the conformal spin coefficients in a general asymptotically flat space-time are given by (see \S9.8 of \cite{Penrose1986}, \cite{Newman1962})
\begin{align}
    \rho - \bar{\rho} &= O(r^{-2}) , & \rho' - \bar{\rho}' &= O(r^{-2}) , & \tau - \bar{\tau}' &= O(r^{-3}) ,
\end{align}
which are the same as the fall-off rates in Kerr. Hence, conformally defined generalized NP constants will satisfy flux balance laws whose flux tends to fall off faster, making such a conformal strategy promising. 

\subsection{Outline of the paper}
The structure of the paper is as follows: in \S\ref{sec:conformal_GHP} I describe the conformal extension to the GHP formalism, which will allow us to cast the familiar Minkowski space NP conservation laws in a manifestly covariant and conformally invariant form in \S\ref{sec:NP_constants_in_conformally_flat_spacetime}. In \S\ref{sec:Aretakis_charge}, I study Aretakis charges and exact conservation laws on null hypersurfaces. Finally, I discuss a straightforward extension of the flat NP constants to spherically symmetric space-times in \S\ref{sec:NP_constants_in_Schwarzschild}. For a short overview of the GHP formalism, refer to Appendix \ref{appendix:GHP}.

\setcounter{theorem}{0}

\section{Conformal extension of the GHP formalism}
\label{sec:conformal_GHP}

In this section I describe a conformal extension to the GHP formalism \cite{Geroch1973,Penrose1985,Penrose1986} which is different from earlier conformal extensions (cf. \cite{Penrose1985,Ludwig1988}). In anticipation of further applications to more general classes of space-times like the Kerr space-time, the formalism put forth here was designed to greatly simplify a number of field equations and commutators in the Kerr space-time. 

\subsection{Conformal weighted thorn and eth}

The NP formalism has several advantages. It achieves enormous notational economy by working with complex numbers. The 24 independent components of the connection are described by 12 complex spin coefficients, and the Weyl curvature is given by 5 complex scalars, which describe the 10 independent components of the Weyl tensor (which has 252 components that are related by complicated symmetries). Furthermore, the spin coefficients come in pairs which get exchanged when exchanging the two real null tetrad vectors, effectively halving the number of equations again. \\
\\
The GHP formalism further improves on the NP formalism by fixing the two real null directions and casting the NP equations into a form which is invariant under a change of tetrad leaving the two real null directions invariant. Four spin coefficients are absorbed into four invariant differential operators. If a space-time contains a pair of null directions with physical significance, the spin coefficients and curvature scalars obtain a clear geometric interpretation. \\
\\
For all of these reasons, it is no surprise that the GHP formalism has proven to be incredibly useful in describing black holes, since they are Petrov type D and thus single out two null directions. Choosing a tetrad along the principal null directions in a vacuum type D space-time leaves only $\rho$, $\rho'$, $\tau$, $\tau'$ and $\Psi_2$ as non-vanishing GHP spin coefficients and curvature scalars. \\
\\
When applied to conformal equations, such as the massless free-field equations (like the Teukolsky equation), there is a further simplification which effectively cuts the number of spin coefficients in half yet again. One may define \textit{conformal GHP operators}
\begin{subequations}
\label{eq:conformal_GHP}
\begin{align}
    \th_c &:= \th + [w - \tfrac{1}{2}(p+q)]\rho - q(\rho - \bar{\rho}) , & \th'_c &:= \th' + [w + \tfrac{1}{2}(p+q)]\rho' + q(\rho' - \bar{\rho}') , \\
    \eth_c &:= \eth + [w - \tfrac{1}{2}(p-q)]\tau + q(\tau - \bar{\tau}') , & \eth'_c &:= \eth' + [w + \tfrac{1}{2}(p-q)]\tau' + q(\bar{\tau} - \tau') . 
\end{align}
\end{subequations}
These conformal operators are conformal weighted operators acting on functions with conformal weight $w$. A function $\eta$ has conformal weight $w$ if under a conformal rescaling of the metric $\hat{g}_{ab} = \Omega^2g_{ab}$, $\hat{\eta} = \Omega^w\eta$. The tetrad is chosen to rescale as
\begin{align}
    (\hat{l}_a, \hat{n}_a, \hat{m}_a) = (\Omega l_a, \Omega n_a, \Omega m_a) ,
\end{align}
although any other rescaling compatible with the metric rescaling $\hat{g}_{ab} = \Omega^2g_{ab}$ could have been considered, for example $\hat{l}_a, \hat{n}_a, \hat{m}_a = \Omega^2l_a, n_a, \Omega m_a$. Since all possible tetrad rescalings are related to each other through a boost, and the GHP formalism is invariant under boosts, the particular choice is not important. The conformal GHP operators all have conformal weight $w=-1$, so that, for example, $\widehat{\th_c\eta} = \Omega^{w-1}\th_c\eta$. The remaining spin coefficients can all be written in terms of the Lie derivative, which is independent of the metric, as
\begin{subequations}
\label{eq:conformal_spin_coefficients}
\begin{align}
    \kappa &= l_a\mathcal{L}_ml^a \\
    \sigma &= m_a\mathcal{L}_ml^a \\
    \rho - \bar{\rho} &= l_a\mathcal{L}_m\bar{m}^a \\
    \tau - \bar{\tau}' &= m_a\mathcal{L}_nl^a , 
\end{align}
\end{subequations}
or alternatively using Cartan's equations (cf. \S4.13 of \cite{Penrose1985})
\begin{subequations}
\begin{align}
    \mathbf{l} \wedge d\mathbf{l} &= (\rho - \bar{\rho})\mathbf{l} \wedge \mathbf{m} \wedge \mathbf{\bar{m}} + \kappa\mathbf{l} \wedge \mathbf{\bar{m}} \wedge \mathbf{n} + \bar{\kappa}\mathbf{l} \wedge \mathbf{m} \wedge \mathbf{n} \\
    \mathbf{m} \wedge d\mathbf{m} &= \sigma\mathbf{m} \wedge \mathbf{\bar{m}} \wedge \mathbf{n} + \bar{\sigma}'\mathbf{m} \wedge \mathbf{\bar{m}} \wedge \mathbf{l} + (\tau - \bar{\tau}')\mathbf{m} \wedge \mathbf{n} \wedge \mathbf{l} , 
\end{align}
\end{subequations}
which are clearly conformal, since, for example, $\hat{\mathbf{l}}\wedge d\hat{\mathbf{l}} = \Omega\mathbf{l}\wedge (\Omega d\mathbf{l} + \mathbf{l}\wedge d\Omega) = \Omega^2\mathbf{l} \wedge d\mathbf{l}$. Hence it is easy to see that the coefficients \eqref{eq:conformal_spin_coefficients} (and their primed variants) are all conformal, with conformal weights $w=-1$. In vacuum type D space-times, the only non-vanishing conformal spin coefficients and curvature scalars are $\rho - \bar{\rho}$, $\rho' - \bar{\rho}'$, $\tau - \bar{\tau}'$ and $\Psi_2$. 
\begin{remark}
    The conformal GHP operators \eqref{eq:conformal_GHP} are far from unique. For example, the operators $\th_c + wf + pg + qh$, where $f$, $g$, and $h$ are weighted scalars with weights $\{w,p,q\} = \{-1,1,1\}$ are also conformal weighted GHP operators. For example, Penrose \cite{Penrose1985} defines
    \begin{align}
        \th_{c,\mathrm{Penrose}} &:= \th + [w - \tfrac{1}{2}(p+q)]\rho , & \eth_{c,\mathrm{Penrose}} &:= \eth + [w - \tfrac{1}{2}(p-q)]\tau ,
    \end{align}
    and Ludwig \cite{Ludwig1988} defines
    \begin{align}
        \th_{c,\mathrm{Ludwig}} &:= \th + \tfrac{1}{2}[w - \tfrac{1}{2}(p+q)](\rho + \bar{\rho}) , & \eth_{c,\mathrm{Ludwig}} &:= \eth + \tfrac{1}{2}[w - \tfrac{1}{2}(p-q)](\tau + \bar{\tau}') ,
    \end{align}
    which have the nice properties $\th_{c,\mathrm{Ludwig}} = \overline{\th_{c,\mathrm{Ludwig}}}$ and $\eth_{c,\mathrm{Ludwig}} = \overline{\eth'_{c,\mathrm{Ludwig}}}$. 
\end{remark}

\noindent The reason for choosing the particular operators \eqref{eq:conformal_GHP} is that it greatly simplifies a large number of field equations and commutators. The GHP equations not involving the Ricci curvature are conformal, and are given by
\begin{subequations}
\begin{align}
    \th_c(\rho - \bar{\rho}) &= \eth_c'\kappa - \eth_c\bar{\kappa} , \\
    \th_c\sigma - \eth_c\kappa &= \Psi_0 \label{eq:psi0} , \\
    \eth_c(\tau - \bar{\tau}') &= \th_c'\sigma - \th_c\bar{\sigma}' .
\end{align}
\end{subequations}
The commutator $[\th_c, \eth_c]$ is given by
\begin{multline}
    \label{eq:commutator}
    [\th_c, \eth_c] = \sigma\eth_c' - \kappa\th_c' - [w - \tfrac{1}{2}(p+q)]\eth_c'\sigma + [w - \tfrac{1}{2}(p-q)]\th_c'\kappa + 2w\Psi_1 \\
    + q\Bigl[\kappa(\rho' - \bar{\rho}') - \sigma(\bar{\tau} - \tau') - \bar{\kappa}\sigma' + (\th_c + \rho - \bar{\rho})(\tau - \bar{\tau}') + (\eth_c + \tau - \bar{\tau}')(\rho - \bar{\rho})\Bigr] .
\end{multline}
The right-hand side only had \textit{two} derivative operators, while $[\th, \eth] = \sigma\eth' - \kappa\th' + \rho\eth - \bar{\tau}'\th + \dots$ has four! The commutators $[\th_c, \th_c']$ and $[\eth_c, \eth_c']$ are more complicated than their GHP counterpart and involve the Ricci curvature. Their difference, however, is simple and is given by
\begin{multline}
\label{eq:wave_commutator}
    [\th_c, \th_c'] - [\eth_c, \eth_c'] = p\Bigl[2\sigma\sigma' - 2\kappa\kappa' - 3\Psi_2\Bigr] + q\Bigl[\th_c'(\rho - \bar{\rho}) + \th_c(\rho' - \bar{\rho}') - \eth_c'(\tau - \bar{\tau}') + \eth_c(\bar{\tau} - \tau') \\
    - (\tau - \bar{\tau}')(\bar{\tau} - \tau') + \sigma\sigma' - \bar{\sigma}\bar{\sigma}' - \kappa\kappa' - \bar{\kappa}\bar{\kappa}' + \eth'\tau + \eth\tau' + \rho\rho' - \bar{\rho}\bar{\rho}' - 2\Phi_{11} - 2\Lambda\Bigr] . 
\end{multline}
The Bianchi identities are not conformal, since they involve the Ricci curvature. Curiously, however, the vacuum Bianchi identities \textit{can} be cast in a conformally invariant form by writing $\psi_k := \Omega^{-1}\Psi_k$, where $\psi_k$ has conformal weight $w = -3$ (cf. \cite{Penrose1985,Penrose1986})
\begin{subequations}
\label{eq:Bianchi}
\begin{align}
    \th_c\psi_{k+1} - \eth_c'\psi_k &= k\sigma'\psi_{k-1} - (3-k)\kappa\psi_{k+2} , \\
    \th_c'\psi_k - \eth_c\psi_{k+1} &= (3-k)\sigma\psi_{k+2} - k\kappa'\psi_{k-1} , 
\end{align}
\end{subequations}
where $0 \leq k \leq 3$.

\subsection{Conformal field equations}

A large number of physically interesting field equations are conformal. Some important examples include: 
\begin{itemize}
    \item Maxwell's source-free equations $\nabla_{[a}F_{bc]} = 0$, $\nabla^aF_{ab} = 0$, where $F_{ab}$ has conformal weight $0$. 
    \item The conformal scalar wave equation $(\Box + \frac{1}{6}R)\phi$, where $\phi$ has conformal weight $-1$. 
    \item The Penrose wave equation $(\Box + \frac{1}{2}R)W_{abcd} = C_{ab}^{\;\;\;ef}W_{cdef} + 4C^e_{\;\;af[c}W^f_{\;\;d]eb}$, where $W_{abcd}$ has conformal weight $1$, and is given by $W_{abcd} := \Omega^{-1}C_{abcd}$. 
\end{itemize}
These equations can be written in conformal GHP form as
\begin{subequations}
\begin{align}
    &\textrm{Maxwell's source-free equations} \quad \th_c\phi_{k+1} - \eth'_c\phi_k = k\sigma'\phi_{k-1} - (1-k)\kappa\phi_{k+2}, \label{eq:Maxwell} \\
    &\textrm{The conformal scalar wave equation} \quad \Bigl[\th_c\th'_c - \eth_c\eth'_c - \Psi_2 + \sigma\sigma' - \kappa\kappa'\Bigr]\phi = 0 ,\label{eq:scalar_wave} \\
    &\textrm{The Penrose wave equation}
    \quad \Bigl[\th_c\th'_c - \eth_c\eth'_c\Bigr]\psi_k = k\Bigl[\eth_c(\sigma'\psi_{k-1}) - \th_c(\kappa'\psi_{k-1})\Bigr] \label{eq:Penrose} \\
    &\quad + \Bigl[\sigma\eth_c' - \kappa\th_c' + (5-k)(\eth_c'\sigma - \th_c'\kappa) - 6\Psi_1\Bigr]\psi_{k+1} + (3-k)\Bigl[\th_c(\sigma\psi_{k+2}) - \eth_c(\kappa\psi_{k+2})\Bigr] , \nonumber
\end{align}
\end{subequations}
together with the primed versions of Eq.~\eqref{eq:Maxwell} and Eq.~\eqref{eq:Penrose}. In the case $k = 0$, the Penrose wave equation can be written as
\begin{align}
    \label{eq:Teukolsky}
    \Bigl[\th_c\th'_c - \eth_c\eth'_c - 3(\Psi_2 + \sigma\sigma' - \kappa\kappa')\Bigr]\psi_0 &= \Bigl[4\sigma\eth_c' - 4\kappa\th_c' + 5(\eth_c'\sigma - \th_c'\kappa) - 6\Psi_1\Bigr]\psi_1 ,
\end{align}
which becomes the Teukolsky equation in Kerr where $\sigma = \kappa = 0 = \Psi_0 = \Psi_1$. Equations \eqref{eq:Bianchi} and \eqref{eq:Maxwell}, and equations \eqref{eq:scalar_wave} and \eqref{eq:Teukolsky} look suspiciously similar. Indeed, \eqref{eq:Bianchi} and \eqref{eq:Maxwell} are examples of the more general \textit{massless free-field equations}
\begin{subequations}
\label{eq:massless_ff}
\begin{align}
     \th_c\phi_{k+1} - \eth'_c\phi_k &= k\sigma'\phi_{k-1} - (2s - 1 - k)\kappa\phi_{k+2} , \\
     \th_c'\phi_k - \eth_c\phi_{k+1} &= (2s - 1 - k)\sigma\phi_{k+2} - k\kappa'\phi_{k-1} , 
\end{align}
\end{subequations}
where $0 \leq k \leq 2s-1$, and where $\phi_k$ corresponds to the electromagnetic NP scalar if it has weights $\{-2, 2(1-k), 0\}$, it corresponds to the rescaled Weyl scalar $\psi_k$ if it has weights $\{-3, 2(2-k), 0\}$, and in general corresponds to a scalar with weights $\{-s-1, 2(s-k), 0\}$ when $\phi_k$ is an NP scalar of a spin-$s$ field. \\
\\
Any field satisfying the massless free-field equations \eqref{eq:massless_ff} also satisfies the spin-$s$ wave equation, generalizing the conformal scalar wave equation \eqref{eq:scalar_wave} and the Penrose wave equation \eqref{eq:Penrose},
\begin{subequations}
\label{eq:massless_wave}
\begin{align}
    \label{eq:massless_wave_a}
    &\Bigl[\th_c\th'_c - \eth_c\eth'_c\Bigr]\phi_k = k\Bigl[\eth_c(\sigma'\phi_{k-1}) - \th_c(\kappa'\phi_{k-1})\Bigr] + \Bigl[\sigma\eth_c' - \kappa\th_c' + (1 + 2s - k)(\eth_c'\sigma - \th_c'\kappa) \nonumber \\
    &\quad - 2(1 + s)\Psi_1\Bigr]\phi_{k+1} + (2s - 1- k)\Bigl[\th_c(\sigma\phi_{k+2}) - \eth_c(\kappa\phi_{k+2})\Bigr] , \quad 0 \leq k \leq 2s - 1 , \\
    &\Bigl[\th_c\th'_c - \eth_c\eth'_c + 2s(2\sigma\sigma' - 2\kappa\kappa' - 3\Psi_2)\Bigr]\phi_{2s} = \Bigl[\sigma'\eth_c - \kappa'\th_c + (1 + 2s)(\eth_c\sigma' - \th_c\kappa') \nonumber \\
    &\quad - 2(1 + s)\Psi_3\Bigr]\phi_{2s-1} + (2s - 1)\Bigl[\th'_c(\sigma'\phi_{2s-2}) - \eth'_c(\kappa'\phi_{2s-2})\Bigr] .
\end{align}
\end{subequations}

\begin{proof}
The scalar $\phi_k$ has weights $\{w,p,q\} = \{-1-s,2(s-k),0\}$, and notably has weight $q=0$. The wave equations follow from the commutators, which simplify significantly when acting on quantities with weight $q=0$: $[\th_c, \eth_c] = \sigma\eth_c' - \kappa\th_c' - [w-\tfrac{1}{2}p](\eth_c'\sigma - \th_c'\kappa) + 2w\Psi_1$ and $[\th_c, \th'_c] - [\eth_c, \eth'_c] = p(2\sigma\sigma' - 2\kappa\kappa' - 3\Psi_2)$. For $0 \leq k \leq 2s-1$, we compute
\begin{align}
    \Bigl[\th_c\th'_c - \eth_c\eth'_c\Bigr]\phi_k &= \th_c\Bigl(\eth_c\phi_{k+1} + (2s-1-k)\sigma\phi_{k+2} - k\kappa'\phi_{k-1}\Bigr) \nonumber \\
    &\qquad - \eth_c\Bigl(\th_c\phi_{k+1} - k\sigma'\phi_{k-1} + (2s-1-k)\kappa\phi_{k+2}\Bigr) \nonumber \\
    &= [\th_c, \eth_c]\phi_{k+1} + k\Bigl[\eth_c(\sigma'\phi_{k-1}) - \th_c(\kappa'\phi_{k-1})\Bigr] \\
    &\qquad + (2s - 1- k)\Bigl[\th_c(\sigma\phi_{k+2}) - \eth_c(\kappa\phi_{k+2})\Bigr] \nonumber \\
    &= k\Bigl[\eth_c(\sigma'\phi_{k-1}) - \th_c(\kappa'\phi_{k-1})\Bigr] + \Bigl[\sigma\eth_c' - \kappa\th_c' + (1 + 2s - k)(\eth_c'\sigma - \th_c'\kappa) \nonumber \\
    &\qquad - 2(1 + s)\Psi_1\Bigr]\phi_{k+1} + (2s - 1- k)\Bigl[\th_c(\sigma\phi_{k+2}) - \eth_c(\kappa\phi_{k+2})\Bigr] . \nonumber
\end{align}
$\phi_{2s}$ satisfies the primed equation
\begin{multline}
    \Bigl[\th_c'\th_c - \eth_c'\eth_c\Bigr]\phi_{2s} = \Bigl[\sigma'\eth_c - \kappa'\th_c + (1 + 2s)(\eth_c\sigma' - \th_c\kappa') - 2(1 + s)\Psi_3\Bigr]\phi_{2s-1} \\
    + (2s - 1)\Bigl[\th_c'(\sigma'\phi_{2s-2}) - \eth_c'(\kappa'\phi_{2s-2})\Bigr] .
\end{multline}
Commuting the derivatives on the left-hand side using Eq.~\eqref{eq:wave_commutator} yields
\begin{multline}
    \Bigl[\th_c\th'_c - \eth_c\eth'_c + 2s(2\sigma\sigma' - 2\kappa\kappa' - 3\Psi_2)\Bigr]\phi_{2s} = \Bigl[\sigma'\eth_c - \kappa'\th_c + (1 + 2s)(\eth_c\sigma' - \th_c\kappa') - 2(1 + s)\Psi_3\Bigr]\phi_{2s-1} \\
    + (2s - 1)\Bigl[\th'_c(\sigma'\phi_{2s-2}) - \eth'_c(\kappa'\phi_{2s-2})\Bigr] .
\end{multline}
\end{proof}

\subsection{Some applications of the conformal GHP formalism}

Stokes' theorem, which we will need in the next section, is independent of the metric altogether and is therefore also (trivially) conformally invariant. The following proposition expresses Stokes' theorem on a null surface in conformal GHP form: 

\begin{proposition}
    Let $n^a$ be a future directed null congruence generating a null hypersurface $\mathcal{N}$. Let $\Sigma$ be a portion of $\mathcal{N}$ having a future boundary $\mathcal{S}'$ and a past boundary $\mathcal{S}$, and let $m^a$ be a complex null vector field tangent to $\mathcal{S}'$ and $\mathcal{S}$. The fundamental theorem of calculus on $\mathcal{N}$ is given by 
    \begin{align}
        \label{eq:conformal_stokes}
        \int_\Sigma\bigl[\th'_c\mu_0 - \eth_c\mu_1 - (\eth'_c + \bar{\tau} - \tau')\nu_1\bigr]\mathcal{N} &= \oint_{\mathcal{S}'}\mu_0\mathcal{S} - \oint_\mathcal{S}\mu_0\mathcal{S} ,
    \end{align}
    where $\mathcal{N}$ and $\mathcal{S}$ denote the volume form. The weighted scalars have weights $\{w, p, q\}$ given by
    \begin{align}
        \mu_0 &: \{-2, 0, 0\} , & \mu_1 &: \{-2, -2, 0\} , & \nu_1 &: \{-2, 0, -2\} ,
    \end{align}
    but are otherwise arbitrary.    
\end{proposition}

\begin{proof}
    The fundamental theorem of calculus $\int_\Sigma d\boldsymbol{\beta} = \oint_{\mathcal{S}'} \boldsymbol{\beta} - \oint_\mathcal{S} \boldsymbol{\beta}$ on a null surface is given in GHP from as (see \S4.14 of \cite{Penrose1985})
    \begin{align}
        \int_\Sigma\bigl[(\th' - 2\rho')\mu_0 - (\eth - \tau)\mu_1 - (\eth' - \bar{\tau})\nu_1\bigr]\mathcal{N} &= \oint_{\mathcal{S}'}\mu_0\mathcal{S} - \oint_\mathcal{S}\mu_0\mathcal{S} ,
    \end{align}
    where $\mu_0 = -2i\beta_{ab}m^a\bar{m}^b$, $\mu_1 = 2i\beta_{ab}\bar{m}^an^b$ and $\nu_1 = -2i\beta_{ab}m^an^b$. If $\beta_{ab}$ has conformal weight $0$, then $\mu_0$, $\mu_1$ and $\nu_1$ all have conformal weight $-2$. The result \eqref{eq:conformal_stokes} follows by re-writing the GHP operators in terms of the conformal GHP operators \eqref{eq:conformal_GHP}. 
\end{proof}

\noindent We obtain a conservation law when we can find appropriately weighted scalars $\mu_0$, $\mu_1$, and $\nu_1$, satisfying
\begin{align}
    \label{eq:conservation_law}
    \th'_c\mu_0 - \eth_c\mu_1 - (\eth'_c + \bar{\tau} - \tau')\nu_1 &= 0
\end{align}
\begin{example}[Conservation of charge]
    In the special case that $\mu_0 = \phi_1$, $\mu_1 = \phi_2$ are electromagnetic scalars we obtain conservation of charge. The component of Maxwell's equations \eqref{eq:Maxwell}$'$ with $k=1$ is given by $\th_c'\phi_1 - \eth_c\phi_2 = -\kappa'\phi_0$. If $n^a$ generates a null surface $\mathcal{N}$ then $\kappa' = 0$ and hence
    \begin{align}
        \oint_{\mathcal{S}'}\phi_1\mathcal{S} - \oint_\mathcal{S}\phi_1\mathcal{S} &= \int_\Sigma(\th_c'\phi_1 - \eth_c\phi_2)\mathcal{N} = 0 .
    \end{align}
    The scalar $\phi_1$ is given by $\phi_1 = -\tfrac{1}{2}(F_{ab} + i(\star F)_{ab})m^a\bar{m}^b$ so that $\oint\phi_1\mathcal{S} = \tfrac{1}{4}\oint\star\mathbf{F} - \tfrac{i}{4}\oint\mathbf{F}$ is indeed simply the electric charge ($-i\times$ the magnetic charge). 
\end{example}

\begin{lemma}[Integration by parts]
    \label{lemma:integration_by_parts}
    Let $\mathcal{S}$ be a closed surface. Then for any weighted functions $f$ and $g$ such that $fg$ has weights $\{-1,-1,1\}$, 
    \begin{align}
        \oint f\eth_cg\mathcal{S} = -\oint g\bigl[\eth_c - (\tau - \bar{\tau}')\bigr]f\mathcal{S} .
    \end{align}
    If $fg$ has weights $\{-1,1,-1\}$, 
    \begin{align}
        \oint f\eth_c'g\mathcal{S} = -\oint g\bigl[\eth_c' - (\bar{\tau} - \tau')\bigr]f\mathcal{S} .
    \end{align}
\end{lemma}

\begin{proof}
    First, compute $\eth_c(fg)$ as
    \begin{align}
        f\eth_cg + g\eth_cf &= \eth_c(fg) = (\eth + \tau - \bar{\tau}')(fg) .
    \end{align}
    Then, using the fact that $\oint\eth\alpha\mathcal{S} = 0$ on a closed surface for any GHP-weighted scalar with weights $\{p,q\} = \{-1,1\}$, both sides can be integrated to obtain
    \begin{align}
        \oint(f\eth_cg + g\eth_cf)\mathcal{S} &= \oint(\tau - \bar{\tau}')fg\mathcal{S} .
    \end{align}
    The integration by parts formula for $\eth_c'$ is obtained in a similar manner, \textit{mutatis mutandis}. 
\end{proof}

The conformal GHP formalism can not only be used to compactly state conformal results, but also efficiently prove them. Before moving on to its application to NP conservation laws in the next section, let me end this section with a wonderfully compact and straightforward proof of the famous Goldberg-Sachs theorem \cite{Goldberg1962} using the conformal GHP formalism. 

\begin{theoremgb}[Goldberg \& Sachs (1962) \cite{Goldberg1962}]
    A nowhere flat vacuum space-time has a multiple gravitational principal null vector $l^a$ if and only if $l^a$ is a shear-free ray. In GHP form this is stated as $\{\Psi_0 = \Psi_1 = 0 \quad \textrm{and} \quad \Psi_p \neq 0\} \Longleftrightarrow \{\kappa = \sigma = 0\}$, for some $2 \leq p \leq 4$.
\end{theoremgb}

\noindent Note that the statement of this theorem is not conformal, since the vacuum condition is not conformal. The proof here, however, is conformal and uses the conformal vacuum Bianchi identities \eqref{eq:Bianchi}. We may therefore weaken the vacuum requirement to merely requiring \textit{conformal vacuum}.\footnote{A space-time with a metric $g_{ab}$ is a conformal vacuum space-time if there exists a conformal factor $\Omega$ such that $\hat{g}_{ab} = \Omega^2g_{ab}$ satisfies the vacuum Einstein equations $\hat{R}_{ab}(\hat{g}) = 0$.} This proof is similar in spirit to Penrose \& Rindler's proof which uses the (non-conformal) GHP formalism---see proposition (7.3.35) of \cite{Penrose1986}. 

\begin{proof}
    Recall that $\psi_k = \Omega^{-1}\Psi_k$. We may set $\Omega = 1$, but beware that $\psi_k$ and $\Psi_k$ have different conformal weights so that making this choice too early would prohibit us from using the conformal GHP formalism. \\
    $\Longrightarrow$: The conformal vacuum Bianchi identities \eqref{eq:Bianchi} imply
    \begin{subequations}
    \begin{align}
        (5-p)\kappa\psi_p &= -\th_c\psi_{p-1} + \eth_c'\psi_{p-2} + (p-2)\sigma'\psi_{p-3} = 0 & \implies \kappa &= 0 , \\
        (5 - p)\sigma\psi_p &= \th_c'\psi_{p-2} - \eth_c\psi_{p-1} + (p - 2)\kappa'\psi_{p-3} = 0 & \implies \sigma &= 0 .
    \end{align}
    \end{subequations}
    $\Longleftarrow$: First, notice that equation \eqref{eq:psi0} immediately implies $\Psi_0 = \th_c\sigma - \eth_c\kappa = 0$. The vacuum Bianchi identities \eqref{eq:Bianchi} imply
    \begin{subequations}
    \begin{align}
        \th_c\psi_1 &= \eth_c'\psi_0 - 3\kappa\psi_2 = 0 , \\
        \eth_c\psi_1 &= \th_c'\psi_0 - 3\sigma\psi_2 = 0 .
    \end{align}
    \end{subequations}
    The commutator \eqref{eq:commutator} then yields
    \begin{align}
        0 = [\th_c, \eth_c]\psi_1 &= \Big[\sigma\eth_c' - \kappa\th_c' + (5-k)(\eth_c'\sigma - \th_c'\kappa) - 6\Psi_1\Bigr]\psi_1 = -6\Psi_1\psi_1 .
    \end{align}
    Hence $\Psi_1 = 0$. 
\end{proof}

\section{Newman-Penrose constants in conformally flat space-time: a covariant and conformal formulation}
\label{sec:NP_constants_in_conformally_flat_spacetime}

The NP conservation laws given by Eq.~\eqref{eq:flat_NP_laws} are strongly coordinate dependent. In this section, in order to later extend the definition to more general wave equations, I will derive a \textit{covariant} formulation of these NP conservation laws for the spin-$s$ wave equation \eqref{eq:massless_wave} in Minkowski space. To achieve this it will be advantageous to use the GHP formalism. This is because the NP constants are defined on the intersection of two light cones. These light cones single out two unique null directions, which define a complex null tetrad up to rotations $m^a \mapsto e^{i\phi}m^a$ and boosts $l^a,n^a \mapsto Al^a,A^{-1}n^a$. Additionally, using the conformal GHP formalism defined in the previous section will further simplify all expressions. \\
\\
I will use Stokes' theorem \eqref{eq:conformal_stokes} as a starting point. To obtain the NP conservation laws we need to specialize to the case when $\mathcal{N}$ is a light cone, $\mathcal{S}$ and $\mathcal{S}'$ are the intersections of light cones, and $\mathcal{N}$, $\mathcal{S}$ and $\mathcal{S}'$ are all spherically symmetric with respect to the same origin. These specializations imply that $\tau - \bar{\tau}' = 0$ and that $l^a$ and $n^a$ are shear-free rays, and hence that the scalars $\phi_k$ satisfy $\th_c\phi_{k+1} = \eth_c'\phi_k$, $\th'_c\phi_k = \eth_c\phi_{k+1}$, and the wave equation (cf. Eq.~\eqref{eq:massless_wave})
\begin{align}
    \label{eq:spin_s_wave}
    \Ro\phi_k := (\th_c\th_c' - \eth_c\eth_c')\phi_k = 0 .
\end{align}
The wave operator $\Ro$ satisfies $\Ro = \Ro'$ when acting on functions with weight $q = 0$ (cf. Eq.~\eqref{eq:wave_commutator}). Consider $\mu_0 = Y\th_c\phi_k$, where the weighted scalar $Y$ is defined as satisfying $\th'_cY = 0 = \eth'_cY$, and having weights $\{s, 2(k-s)-1, -1\}$. We then compute, using the wave equation \eqref{eq:spin_s_wave}$^\prime$,
\begin{align}
    \th'_c\mu_0 &= Y\th'_c\th_c\phi_k = Y\eth'_c\eth_c\phi_k = \eth'_c(Y\eth_c\phi_k) .
\end{align}
Setting $\nu_1 = Y\eth_c\phi_k$ we recognize a conservation law. The quantities $\oint_\mathcal{S}Y\th_c\phi_k\mathcal{S}$ are conserved! Note that $Y$ is a spherical harmonic with $\ell = s - k$. In particular, when $k > s$, $Y = 0$. Indeed, all of this is a more abstract way of observing that the lowest $\ell$-mode ${}_sY_{\ell=s}$ of any harmonic with non-negative spin weight $s$ satisfies $\eth'\eth\;{}_sY_{\ell=s} = 0$, so that $\th'_c\th_c\phi_{k,\ell=s-k} = 0$, where $\phi_{k,\ell}$ denotes the $\ell$-mode of $\phi_k$. \\
\\
The key to finding the NP constants for higher $\ell$-modes is to notice that, if $s \geq 1$, the conserved quantity $\oint_\mathcal{S}Y\th_c\phi_0\mathcal{S}$ can be written in terms of $\phi_1$ as follows: define $Z$ as satisfying $\eth'_cZ = Y$. Then
\begin{align}
\begin{split}
    \label{eq:NP_identity}
    \th_c'(Y\th_c\phi_0) &= \th_c'(\eth_c'(Z\th_c\phi_0) - Z\eth_c'\th_c\phi_0) = \eth_c'\th_c'(Z\th_c\phi_0) - \th_c'(Z\th_c\eth_c'\phi_0) \\
    &= \eth_c'\th_c'(Z\th_c\phi_0) - \th_c'(Z\th_c^2\phi_1) ,
\end{split}
\end{align}
where I used the commutators $[\th_c, \eth_c'] = 0 = [\th_c', \eth_c']$ (cf. Eqs. $\overline{\eqref{eq:commutator}}$ and \eqref{eq:commutator}$'$). It follows that $\oint Z\th_c^2\phi_1\mathcal{S}$ is conserved. If it were possible to define ``$\phi_{-1}$" satisfying ``$\th_c\phi_0 = \eth'_c\phi_{-1}$" and ``$\th'_c\phi_{-1} = \eth_c\phi_0$", and more generally $\phi_k$ satisfying the massless free-field equations \eqref{eq:massless_ff} for \textit{any} integer $k$, then for each $\ell$ and each $k$, there exist conserved quantities of the form $\oint Y\th_c^{\ell-s+k+1}\phi_k\mathcal{S}$ for an appropriately chosen $Y$. However, in general, there cannot be such a $\phi_{-1}$. This can be seen by decomposing $\phi_0$ and $\phi_{-1}$ into spherical harmonics. The lowest $\ell$-mode of $\phi_0$ is $\ell=s$ since $\phi_0$ has spin weight $s$, but the lowest $\ell$-mode of $\phi_{-1}$ is $s+1$. Hence, ``$\th_c\phi_0 = \eth'_c\phi_{-1}$" has solutions $\phi_{-1}$ if and only if $\th_c\phi_0$ has a vanishing $\ell=s$ mode. Instead, let us define $\phi_{-1} := \eth_c\th_c\phi_0$. 
\begin{remark}
    The definition $\phi_{-1} := \eth_c\th_c\phi_0$ is similar in spirit to ``$\th_c\phi_0 = \eth_c'\phi_{-1}$", which implies ``$\eth_c\th_c\phi_0 = \eth_c\eth_c'\phi_{-1}$" so that in terms of spherical harmonics, ``$(\eth_c\th_c\phi_0)_\ell \propto \phi_{-1,\ell}$" for all $\ell$. 
\end{remark}
\noindent $\phi_{-1}$ satisfies the wave equation: 
\begin{align}
\begin{split}
    \Ro\phi_{-1} &= \Ro\eth_c\th_c\phi_0 = (\th_c\th_c' - \eth_c\eth_c')\eth_c\th_c\phi_0 \\
    &= \eth_c\th_c(\th_c'\th_c - \eth_c'\eth_c)\phi_0 = \eth_c\th_c\Ro'\phi_0 = 0 . \label{eq:raising_operator_commutes}
\end{split}
\end{align}
It follows from Eq.~\eqref{eq:raising_operator_commutes} that $\oint_\mathcal{S}Y\th_c\phi_{-1}\mathcal{S}$ is conserved, where $Y$ satisfies $\th'_cY = 0 = \eth'_cY$ and has weights $\{s+1, 2(k-s)-3, 1\}$. This can be written in terms of $\phi_0$ as $-\oint_\mathcal{S}\eth_cY\th_c^2\phi_0$, where in terms of spherical harmonics, $\eth_cY$ consists of ${}_{k-s}Y_{\ell=s+1}$. \\
\\
To generalize this to any $\ell$-mode we use the following important result which follows from Eq.~\eqref{eq:raising_operator_commutes}:

\begin{lemma}
When acting on quantities with weights $\{w, p, 0\}$, 
\begin{align}
    [\Ro, \eth_c\th_c] &= 0 .
\end{align}
\end{lemma}

\noindent Hence, when applying $\eth_c\th_c$ to a function with equal spin- and boost weight satisfying the spin-$s$ wave equation, we can generate new solutions with higher spin. The following conservation laws are a consequence: 

\begin{theorem}
    \label{thm:covariant_NP_conservation_law}
    Let $\phi_k$ be a scalar with weights $\{-s-1, 2(s-k), 0\}$ satisfying $\Ro\phi_k = 0$. Let $\mathcal{N}$ be a light cone in Minkowski space generated by $n^a$. Choose a foliation of $\mathcal{N}$ by spherical cuts $\mathcal{S}$ which are spherically symmetric with respect to a common centre. If $\mathcal{S}$ and $\mathcal{S}'$ are any two cuts belonging to this foliation, then
    \begin{align}
        \oint_\mathcal{S}Y\th_c^{\ell-s+k+1}\phi_k\mathcal{S} &= \oint_{\mathcal{S}'}Y\th_c^{\ell-s+k+1}\phi_k\mathcal{S} , 
    \end{align}
    where $Y = \eth_c^{\ell-s+k}Z$ is a weighted function with weights
    \begin{align}
        Y &: \{\ell+k,-\ell-s+k-1,-\ell+s-k-1\} . \label{eq:Y_weights} 
    \end{align}
    $Z$ satisfies $\th_c'Z = 0 = \eth_c'Z$. In terms of spin-weighted spherical harmonics, $Y$ consists of ${}_{k-s}Y_\ell$. 
\end{theorem}

\begin{proof}
Since $\Ro$ commutes with $\eth_c\th_c$, $(\eth_c\th_c)^{\ell-s+k}\phi_k$ satisfies the wave equation
 \begin{align}
    \Ro(\eth_c\th_c)^{\ell-s+k}\phi_k &= (\eth_c\th_c)^{\ell-s+k}\Ro\phi_k = 0 .
\end{align}
Let $Z$ be a function with weights $\{2\ell - s + 2k, -2\ell-1, -1\}$ satisfying $\th_c'Z = 0 = \eth_c'Z$. Then
\begin{align}
    \th_c'\bigl[Z\eth_c^{\ell-s+k}\th_c^{\ell-s+k+1}\phi_k\bigr] = Z(\Ro' + \eth'_c\eth_c)(\eth_c\th_c)^{\ell-s+k}\phi_k = \eth_c'\bigl[Z\eth_c(\eth_c\th_c)^{\ell-s+k}\phi_k\bigr] .
\end{align}
Using Stokes' theorem \eqref{eq:conformal_stokes} we obtain
\begin{align}
    \oint_\mathcal{S}Z\eth_c^{\ell-s+k}\th_c^{\ell-s+k+1}\phi_k\mathcal{S} = \oint_{\mathcal{S}'}Z\eth_c^{\ell-s+k}\th_c^{\ell-s+k+1}\phi_k\mathcal{S} .
\end{align}
The theorem follows by applying integration by parts to move $\eth_c^{\ell-s+k}$ onto $Z$. 
\end{proof}

\noindent These conservation laws lead us to a manifestly covariant and conformally invariant definition of conserved NP constants for solutions to the spin-$s$ wave equation:

\begin{definition}[Newman-Penrose constants]
    \label{def:NP_constants}
    Let $\phi_k$ be a massless field satisfying the spin-$s$ wave equation $\Ro\phi_k = 0$. Define the Newman-Penrose constants $Q_\ell$ associated to $\phi_k$ on any light cone $\mathcal{C}$ by
    \begin{align}
        \label{eq:NP_constants_definition}
        Q_\ell[\phi_k, \mathcal{C}] &:= \oint_\mathcal{S}Y\th_c^{\ell-s+k+1}\phi_k\mathcal{S} , 
    \end{align}
    where $Y$ is as in theorem \ref{thm:covariant_NP_conservation_law}. They are independent of the spherical cut $\mathcal{S}$. 
\end{definition}

\noindent Concretely, choose the standard outgoing and ingoing null coordinates $u$ and $v$ orthogonal to $\mathcal{S}$, so that the metric is $ds^2 = 4dudv-r^2dS^2$. On $\mathcal{C} = \{v = \mathrm{constant}\}$ the NP constants are given by
\begin{align}
    Q_\ell &= \oint_\mathcal{S}{}_{k-s}Y_\ell r^{-2\ell}\partial_v(r^2\partial_v)^{\ell-s+k}(r^{2s-k+1}\phi_k) dS .
\end{align}
They satisfy $\partial_uQ_\ell = 0$. When $s = k = 0$, we recover the usual expression for the NP constants for the scalar wave equation Eq.~\eqref{eq:flat_NP_laws}.

\section{Aretakis charges}
\label{sec:Aretakis_charge}

The question of whether or not there exist any functions $Y$ on a curved space-time such that the constants defined by Eq.~\eqref{eq:NP_constants_definition} are conserved was investigated by Aretakis in \cite{Aretakis2013}, in the special case $s = k = \ell = 0$ for the scalar wave equation $\Box\phi = 2(\Ro - \Psi_2 - 2\Lambda)\phi = 0$. It was found that $Q_{\ell=0}[\phi]$ is conserved if $Y$ is the solution to a certain elliptic differential equation on $\mathcal{N}$ (see also \cite{Aretakis2014}). Using the conformal GHP formalism, it is straightforward to generate a similar statement for wave equations of arbitrary spin: 

\begin{theorem}
\label{thm:elliptic_operator}
Associated to solutions $\phi_k$ to the spin-$s$ wave equation $(\Ro - V)\phi_k = 0$, where $V$ is an arbitrary function, there exist conserved quantities on a null hypersurface $\mathcal{N}$ generated by $l^a$ of the form
\begin{align}
    \oint\omega\th_c'\phi_k\mathcal{S}
\end{align}
if and only if there exists a function $\omega$ with weights $\{s, 2(k-s)+1, 1\}$ satisfying $\th_c\omega = 0$ and the elliptic differential equation
\begin{align}
\label{eq:elliptic_pde2}
    \bigl[\eth_c'(\eth_c - (\tau - \bar{\tau}')) + V\bigr]\omega &= 0 .
\end{align}
\end{theorem}

\begin{proof}
The proof is a simple application of Stokes' theorem \eqref{eq:conformal_stokes}: 
\begin{align}
    &\th_c(\omega\th_c'\phi_k) = \th_c\omega\th_c'\phi_k + \omega(\eth_c\eth_c' + V)\phi_k \nonumber \\
    &\qquad\qquad\; = \th_c\omega\th_c'\phi_k + (\eth_c - (\tau - \bar{\tau}'))(\omega\eth_c'\phi_k) - (\eth_c - (\tau - \bar{\tau}'))\omega\eth_c'\phi_k + V\omega\phi_k \\
    &= \th_c\omega\th_c'\phi_k + (\eth_c - (\tau - \bar{\tau}'))(\omega\eth_c'\phi_k) - \eth_c'\bigl[\phi_k(\eth_c - (\tau - \bar{\tau}'))\omega\bigr] + \phi_k[\eth_c'(\eth_c - (\tau - \bar{\tau}')) + V]\omega . \nonumber
\end{align}
Hence, if $\th_c\omega = 0 = [\eth_c'(\eth_c - (\tau - \bar{\tau}')) + V]\omega$, we obtain the conservation law (cf. Eq.~\eqref{eq:conservation_law})
\begin{align}
    \th_c\mu_0 - \eth_c'\mu_1 - (\eth_c - (\tau - \bar{\tau}'))\nu_1 &= 0 ,
\end{align}
where
\begin{align}
    \mu_0 &= \omega\th_c'\phi_k , & \mu_1 &= -\phi_k(\eth_c - (\tau - \bar{\tau}'))\omega , & \nu_1 &= \omega\eth_c'\phi_k . 
\end{align}
\end{proof}

\noindent Aretakis discovered that on extremal Killing horizons, Eq.~\eqref{eq:elliptic_pde2} has a non-trivial solution in the case $s = 0$, $V = \Psi_2 + 2\Lambda$. 

\begin{theorem}[Aretakis \cite{Aretakis2013}]
\label{thm:Aretakis2013}
Let $\phi$ be a solution to the wave equation $\Box\phi = 0$, and let $\mathcal{H}$ be an extremal Killing horizon---a null hypersurface generated by a Killing vector field $k^a \overset{\mathcal{H}}{=} kl^a$ satisfying the extremality condition $k^b\nabla_bk^a \overset{\mathcal{H}}{=} 0$. Let $\mathcal{H}$ be foliated by closed space-like cross-sections $\mathcal{S}$ whose intrinsic metric $2m_{(a}\bar{m}_{b)}$ satisfies $\mathcal{L}_k(2m_{(a}\bar{m}_{b)}) \overset{\mathcal{H}}{=} 0$. Then there exists a positive function $\omega$ on $\mathcal{H}$, which is unique up to a factor, such that
\begin{align}
    \oint_\mathcal{S}\omega\th_c'\phi\mathcal{S}
\end{align}
is independent of the cross section $\mathcal{S}$. 
\end{theorem}

\noindent Aretakis proved this using a clever choice of coordinates. It is also instructive to see how this can be proven using the conformal GHP formalism. 

\begin{proof}
In conformal GHP form, 
\begin{align}
    \Box\phi = 2(\Ro - \Psi_2 - 2\Lambda)\phi = 0 .
\end{align}
Let $k^a$ be a Killing vector such that $k^a \overset{\mathcal{H}}{=} kl^a$. All but one of Killing's equations are conformal, the most relevant of which can be written as \cite{Kolassis1993}
\begin{align}
    \sigma &\overset{\mathcal{H}}{=} 0 , & (\eth_c - (\tau - \bar{\tau}'))k &\overset{\mathcal{H}}{=} 0 , & \eth_c'k &\overset{\mathcal{H}}{=} 0 , \label{eq:conformal_Killing}
\end{align}
where the scalar $k$ has weights $\{1, 1, 1\}$. These are the \textit{conformal Killing equations} $\nabla_{(a}k_{b)} \propto g_{ab}$. If in addition to satisfying the conformal Killing equations $k^a$ satisfies $0 = \nabla_ak^a \overset{\mathcal{H}}{=} (\th - 2\rho)k$, then $k^a$ is a Killing vector satisfying $\nabla_{(a}k_{b)} = 0$. \\
\\
Let $0 = \mathcal{L}_k\omega = k\th_c\omega$. Define the elliptic operator $\mathcal{O}$ as
\begin{align}
    \mathcal{O}\omega := k[\eth'_c(\eth_c - (\tau - \bar{\tau}')) + \Psi_2 + 2\Lambda]\omega = k[\eth'\eth + \Psi_2 + 2\Lambda]\omega ,
\end{align}
which is just the operator defined in theorem \ref{thm:elliptic_operator} multiplied by $k$. The operator $\mathcal{O}$ is \textit{real}. To see this, apply the commutator $[\eth, \eth']\omega = (\Psi_2 - \bar{\Psi}_2)\omega$ (cf. \eqref{eq:eth_commutator}) to compute
\begin{align}
    \mathcal{O}\omega &= k[\eth'\eth + \Psi_2 + 2\Lambda]\omega = k[\eth\eth' + \bar{\Psi}_2 + 2\Lambda]\omega = \overline{\mathcal{O}\bar{\omega}} .
\end{align}
Hence, by standard elliptic PDE theory (see, for example, \S6 of \cite{Evans1998}) the principal (maximal) eigenvalue $\lambda$ of $\mathcal{O}$ is real, and we can take $\omega$ to be the unique (up to a factor) positive principal eigenfunction. \\
\\
Next, note that with our choice of $m^a$, 
\begin{align}
    0 \overset{\mathcal{H}}{=} m^b\mathcal{L}_k(2m_{(a}\bar{m}_{b)}) = m^b\mathcal{L}_k(2l_{(a}n_{b)} - g_{ab}) = l_am^b\mathcal{L}_kn_b + n_am^b\mathcal{L}_kl_b \overset{\mathcal{H}}{=} l_am_b\mathcal{L}_kn^b ,
\end{align}
so that $0 \overset{\mathcal{H}}{=} m_a\mathcal{L}_ln^a = -(\tau - \bar{\tau}')$ (cf. \eqref{eq:conformal_spin_coefficients}). It follows that $\eth k \overset{\mathcal{H}}{=} -\bar{\tau}'k$ and $\eth'k \overset{\mathcal{H}}{=} -\tau'k$ (cf. \eqref{eq:conformal_Killing}), with which we compute (cf. \eqref{eq:spin_coeff_f})
\begin{align}
\begin{split}
    \lambda\oint_\mathcal{S}\omega\mathcal{S} &= \oint_\mathcal{S}\mathcal{O}\omega\mathcal{S} = \oint_\mathcal{S}[\eth\eth'k + k\Psi_2 + 2k\Lambda]\omega\mathcal{S} = \oint_\mathcal{S}[\eth(-\tau'k) + k\Psi_2 + 2k\Lambda]\omega\mathcal{S} \\
    &= \oint_\mathcal{S}k[-\eth\tau' + \tau'\bar{\tau}' + \Psi_2 + 2\Lambda]\omega\mathcal{S} = -\oint_\mathcal{S}k\omega(\th - \rho)\rho'\mathcal{S} .
\end{split}
\end{align}
The surface gravity $\kappa_g$ is defined as $\kappa_gk^a \overset{\mathcal{H}}{=} k^b\nabla_bk^a \overset{\mathcal{H}}{=} \th kk^a$. If $k^a$ is Killing, then $0 = \nabla_ak^a \overset{\mathcal{H}}{=} \th k - 2\rho k$, so that $\kappa_g = \th k = 2\rho k$. Furthermore, since $0 \overset{\mathcal{H}}{=} \mathcal{L}_k(k\rho') \overset{\mathcal{H}}{=} k\th(k\rho') \overset{\mathcal{H}}{=} k(k\th\rho' + \kappa_g\rho')$, it follows that
\begin{align}
    \lambda\oint_\mathcal{S}\omega\mathcal{S} &= \tfrac{3}{2}\oint_\mathcal{S}\omega\kappa_g\rho'\mathcal{S} .
\end{align}
If $\kappa_g = 0$, then $\lambda = 0$ and hence $\omega$ is a solution to $\mathcal{O}\omega = 0$. Theorem \ref{thm:Aretakis2013} now follows from theorem \ref{thm:elliptic_operator}. 
\end{proof}

\section{Newman-Penrose-like approximate conservation laws on shear-free null hypersurfaces}
\label{sec:NP_constants_in_Schwarzschild}

The conservation laws defined in \S\ref{sec:NP_constants_in_conformally_flat_spacetime} hold in every conformally flat space-time. This class includes very few physically relevant space-times since the full Weyl tensor is required to vanish. If the Weyl tensor is non-vanishing, Huygens' principle does not hold and we should not expect similar conservation laws to exist. In this section I derive approximate conservation laws that aim to generalize the NP conservation laws, and which may be integrated to obtain the asymptotics of solutions to wave equations of the form $(\Ro - V)\phi_k = 0$ in spherically symmetric space-times. The strategy will be to promote definition \ref{def:NP_constants} to a definition for generalized NP constants. The conservation theorem \ref{thm:covariant_NP_conservation_law} will be replaced by a flux-balance law, whose flux will decay rapidly under suitable conditions. 

\subsection{NP flux-balance laws in full generality}

Let $Q_\ell$ be defined as 
\begin{align}
    \label{eq:general_NP_consts}
    Q_\ell[\phi_k,\mathcal{S}] &:= \oint_\mathcal{S}Z\th_c(\eth_c\th_c)^{\ell-s+k}\phi_k\mathcal{S} , 
\end{align}
where $\mathcal{S}$ is now allowed to be \textit{any} space-like cut of a generic null hypersurface $\mathcal{N}$ generated by $n^a$, and where $Z$ is an appropriately weighted scalar (with weights as in theorem \ref{thm:covariant_NP_conservation_law}) satisfying $\th_c'Z = 0 = (\eth_c' + \bar{\tau} - \tau')Z$. Under the assumptions of theorem \ref{thm:covariant_NP_conservation_law}, Eq.~\eqref{eq:general_NP_consts} coincides with definition \ref{def:NP_constants}. The first major obstruction that presents itself is the integrability condition
\begin{align}
\label{eq:integrability_condition}
\begin{split}
    0 &= (\th_c'(\eth_c' + \bar{\tau} - \tau') - \eth_c'\th_c')Z = ([\th_c', \eth_c'] + \th_c'(\bar{\tau} - \tau'))Z \\
    &= \bigl[\sigma'(\eth_c + \tau -\bar{\tau}') - \kappa'\th_c - [\ell - s + 2k - 1](\eth_c\sigma' - \th_c\kappa') + 2(2\ell - s + 2k)\Psi_3\bigr]Z ,
\end{split}
\end{align}
which, because of the Goldberg-Sachs theorem, is satisfied when $\mathcal{N}$ is generated by shear-free rays, but not in general. \\
\\
Using Stokes' theorem \eqref{eq:conformal_stokes} we find the following flux-balance laws on $\mathcal{N}$: 
\begin{align}
    \label{eq:NP_flux_balance}
    Q_\ell[\phi_k, \mathcal{S}'] - Q_\ell[\phi_k, \mathcal{S}] &= \int_\Sigma\bigl[\th_c'(Z\th_c(\eth_c\th_c)^{\ell-s+k}\phi_k) - (\eth_c' + \bar{\tau} - \tau')(Z\eth_c(\eth_c\th_c)^{\ell-s+k}\phi_k)\bigr]\mathcal{N} \nonumber \\
    &= \int_\Sigma Z\Ro'(\eth_c\th_c)^{\ell-s+k}\phi_k\mathcal{N} , 
\end{align}
where $\Sigma \subset \mathcal{N}$ has future- and past boundaries $\mathcal{S}'$ and $\mathcal{S}$. Eq.~\eqref{eq:NP_flux_balance} holds for any function $\phi_k$ with weights $\{w,p,q\} = \{-s-1,2(s-k),0\}$. The commutator $[\Ro', \eth_c\th_c]$ generally does not vanish and involves the Weyl tensor. 

\subsection{NP flux-balance laws in spherical symmetry}
\label{sec:Kehrberger_charge}

Naturally, I will be using a tetrad adapted to the spherical symmetry, so that $l^a$ and $n^a$ are in- and outgoing null vectors orthogonal to the spheres of symmetry. The only non-zero GHP spin coefficients are $\rho$ and $\rho'$, which are real so that all conformal GHP spin coefficients are vanishing. Note that $\eth_c = \eth$ in this tetrad, but I will continue to write $\eth_c$ to make it clear that I am working with conformal quantities. I will be considering field components $\phi_k$, with weights $\{-s-1, 2(s-k), 0\}$, satisfying the spin-$s$ wave equation (cf. \eqref{eq:scalar_wave}, \eqref{eq:Teukolsky}, \eqref{eq:massless_wave})
\begin{align}
    (\Ro - V)\phi_k &= 0 . 
\end{align}
The NP flux-balance laws \eqref{eq:NP_flux_balance} are given by
\begin{align}
\begin{split}
    \label{eq:NP_flux_balance_2}
    Q_\ell[\phi_k, \mathcal{S}'] - Q_\ell[\phi_k, \mathcal{S}] &= \int_\Sigma Z\Ro'(\eth_c\th_c)^{\ell-s+k}\phi_k\mathcal{N} \\
    &= \int_\Sigma Z\Bigl([\Ro', (\eth_c\th_c)^{\ell-s+k}] - (\eth_c\th_c)^{\ell-s+k}(\Ro - \Ro' - V)\Bigr)\phi_k\mathcal{N} .
\end{split}
\end{align}
The commutators $[\Ro', \eth_c\th_c]$ and $\Ro - \Ro' = [\th_c, \th_c'] - [\eth_c, \eth_c']$ do not vanish, and involve the curvature. When acting on scalars with weights $\{w, p, 0\}$ the former is simply
\begin{align}
    \label{eq:Ro_commutator}
    \Ro - \Ro' &= -3p\Psi_2 ,
\end{align}
(cf. Eq.~\eqref{eq:wave_commutator}), while the latter can be found by first noting that since $[\th_c, \eth_c] = [\th_c, \eth_c'] = [\th_c', \eth_c] = [\th_c', \eth_c'] = 0$,
\begin{align}
    &\Ro\eth_c\th_c - \eth_c\th_c\Ro' = \th_c\th_c'\eth_c\th_c - \eth_c\eth_c'\eth_c\th_c - \eth_c\th_c\th_c'\th_c + \eth_c\th_c\eth_c'\eth_c = 0 ,
\end{align}
so that
\begin{align}
    \label{eq:Ro_ladder_commutator}
    [\Ro', \eth_c\th_c] &= (\Ro' - \Ro)\eth_c\th_c = 3(p+2)\Psi_2\eth_c\th_c .
\end{align}
Putting everything together, we obtain the following flux-balance laws: 

\begin{theorem}
\label{thm:spherical_NP_flux_balance}
In spherical symmetry, the NP flux-balance laws \eqref{eq:NP_flux_balance} reduce to
\begin{multline}
    \label{eq:NP_flux_balance_spherical}
    Q_\ell[\phi_k, \mathcal{S}'] - Q_\ell[\phi_k, \mathcal{S}] \\
    = \int_\Sigma Y\Bigl(\th_c^{\ell-s+k}(V\phi_k) + \sum_{n=0}^{\ell-s+k}6(n+s-k)\th_c^{\ell-s+k-n}\bigl[\Psi_2\th_c^n\phi_k\bigr]\Bigr)\mathcal{N} ,
\end{multline}
where $Y = (-\eth_c)^{\ell-s+k}Z$ has weights $\{\ell+k,-\ell-s+k-1,-\ell+s-k-1\}$. In terms of spin-weighted spherical harmonics, $Y$ consists of ${}_{k-s}Y_\ell$. 
\end{theorem}

\begin{proof}
From the integrand of Eq.~\eqref{eq:NP_flux_balance_2} we compute, using the commutators \eqref{eq:Ro_commutator} and \eqref{eq:Ro_ladder_commutator}, 
\begin{align}
\begin{split}
    &\Bigl([\Ro', (\eth_c\th_c)^{\ell-s+k}] - (\eth_c\th_c)^{\ell-s+k}(\Ro - \Ro' - V)\Bigr)\phi_k \\
    &= \Bigl(\sum_{n=1}^{\ell-s+k}(\eth_c\th_c)^{\ell-s+k-n}[\Ro', (\eth_c\th_c)](\eth_c\th_c)^{n-1} + (\eth_c\th_c)^{\ell-s+k}(6(s-k)\Psi_2 + V)\Bigr)\phi_k \\
    &= \Bigl(\sum_{n=1}^{\ell-s+k}6(n+s-k)(\eth_c\th_c)^{\ell-s+k-n}\bigl[\Psi_2(\eth_c\th_c)^n\bigr] + (\eth_c\th_c)^{\ell-s+k}(6(s-k)\Psi_2 + V)\Bigr)\phi_k \\
    &= \sum_{n=0}^{\ell-s+k}6(n+s-k)(\eth_c\th_c)^{\ell-s+k-n}\bigl[\Psi_2(\eth_c\th_c)^n\phi_k\bigr] + (\eth_c\th_c)^{\ell-s+k}(V\phi_k) .
\end{split}
\end{align}
Inserting this back into \eqref{eq:NP_flux_balance_2} and integrating by parts $\ell-s+k$ times to move $\eth_c^{\ell-s+k}$ onto $Z$, we obtain \eqref{eq:NP_flux_balance_spherical}: 
\begin{multline}
    \label{eq:sph_appr_NP_law}
    Q_\ell[\phi_k, \mathcal{S}'] - Q_\ell[\phi_k, \mathcal{S}] \\
    = \int_\Sigma Z\Bigl(\sum_{n=0}^{\ell-s+k}6(n+s-k)(\eth_c\th_c)^{\ell-s+k-n}\bigl[\Psi_2(\eth_c\th_c)^n\phi_k\bigr] + (\eth_c\th_c)^{\ell-s+k}(V\phi_k)\Bigr)\mathcal{N} \\
    = \int_\Sigma(-\eth_c)^{\ell-s+k}Z\Bigl(\sum_{n=0}^{\ell-s+k}6(n+s-k)\th_c^{\ell-s+k-n}\bigl[\Psi_2\th_c^n\phi_k\bigr] + \th_c^{\ell-s+k}(V\phi_k)\Bigr)\mathcal{N} .
\end{multline}
\end{proof}

\subsection{Coordinate expressions}

To end this section, I will express Eqs. \eqref{eq:general_NP_consts} and \eqref{eq:sph_appr_NP_law} in coordinates. I compute a small sample of NP flux-balance laws in Schwarzschild, which may be compared with, for example, similar expressions appearing in Kehrberger's work \cite{Kehrberger2021a,Kehrberger2021c}. \\
\\
Let us consider the metric in double null coordinates
\begin{align}
    \label{eq:spherical_metric}
    ds^2 = 4fdudv - r^2dS^2 ,
\end{align}
and choose the tetrad $l^a\partial_a := D = \tfrac{1}{2}\partial_v$, $n^a\partial_a := D' = f^{-1}\partial_u$. The non-vanishing spin coefficients are given by
\begin{align}
    \rho &= -r^{-1}Dr , & \rho' &= -r^{-1}D'r , & \epsilon &= \tfrac{1}{2}f^{-1}Df ,
\end{align}
from which the conformal GHP operators can be computed: 
\begin{subequations}
\begin{align}
    \th_c &= D - 2b\epsilon + [w-b]\rho = D - bf^{-1}Df - [w-b]r^{-1}Dr , \\
    \th_c' &= D' + [w+b]\rho' = D - [w+b]r^{-1}D'r , 
\end{align}
\end{subequations}
where $b = \tfrac{1}{2}(p + q)$, which may conveniently be expressed as
\begin{subequations}
\begin{align}
    \th_c\eta &= r^{w-b}f^bD(r^{b-w}f^{-b}\eta) = \tfrac{1}{2}r^{w-b}f^b\partial_v(r^{b-w}f^{-b}\eta) , \\
    \th'_c\eta &= r^{w+b}D'(r^{-w-b}\eta) = r^{w+b}f^{-1}\partial_u(r^{-w-b}\eta) .
\end{align}
\end{subequations}
The function $Y$, as defined by theorem \ref{thm:spherical_NP_flux_balance}, has weight $w+b = k-1$ so that 
\begin{align}
    \th_c'Y = r^{k-1}f^{-1}\partial_u(r^{1-k}Y) = 0 ,
\end{align}
and hence
\begin{align}
    Y = 2^{\ell-s+k+1}r^{k-1}{}_{k-s}Y_\ell ,
\end{align}
where the numerical factor was chosen for convenience. The generalized NP constants \eqref{eq:general_NP_consts} are given by
\begin{align}
    Q_\ell &= \oint_\mathcal{S}{}_{k-s}Y_\ell r^{-2\ell}f^\ell\partial_v(r^2f^{-1}\partial_v)^{\ell-s+k}(r^{2s-k+1}f^{k-s}\phi_k) dS .
\end{align}
The flux-balance laws \eqref{eq:NP_flux_balance_spherical} are
\begin{multline}
    Q_\ell[\phi_k, \mathcal{S}'] - Q_\ell[\phi_k, \mathcal{S}] = \int_\Sigma2{}_{k-s}Y_\ell r^{-2\ell-2}\Bigl((r^2f^{-1}\partial_v)^{\ell-s+k}(r^{2s-k+3}V\phi_k) \\
    + \sum_{n=0}^{\ell-s+k}6(n+s-k)(r^2f^{-1}\partial_v)^{\ell-s+k-n}\bigl[\Psi_2r^{2-n}f^{s-k+n}(r^2f^{-1}\partial_v)^n(r^{2s-k+1}\phi_k)\bigr]\Bigr)dudS ,
\end{multline}
which are given in differential form as
\begin{multline}
    \label{eq:NP_laws_diff_form}
    \partial_uQ_\ell[\phi_k] = 2r^{-2\ell-2}\Bigl((r^2f^{-1}\partial_v)^{\ell-s+k}(r^{2s-k+3}(V\phi_k)_\ell) \\
    + \sum_{n=0}^{\ell-s+k}6(n+s-k)(r^2f^{-1}\partial_v)^{\ell-s+k-n}\bigl[\Psi_2r^{2-n}f^{s-k+n}(r^2f^{-1}\partial_v)^n(r^{2s-k+1}\phi_{k,\ell})\bigr]\Bigr) .
\end{multline}
Finally, as a consistency check, I will generate a few simple examples of Eq.~\eqref{eq:NP_laws_diff_form} which may be compared to the literature. 

\begin{example}[The $\ell = 0$ mode in Schwarzschild]
    Let $\phi = \phi(u,v)$ be a solution to the scalar wave equation on Schwarzschild supported on the $\ell = 0$ mode. Then
    \begin{align}
        \th'_c\th_c\phi - \Psi_2\phi &= 0 .
    \end{align}
    The Schwarzschild metric in double null coordinates is given by Eq.~\eqref{eq:spherical_metric} with $f = 1 - 2mr^{-1}$. The Weyl curvature scalar is given by $\Psi_2 = -mr^{-3}$. From this, we can easily compute (cf. \eqref{eq:Schw_spherical_mode})
    \begin{align}
        \th_c\th'_c\phi - \Psi_2\phi &= \tfrac{1}{2}r^{-1}f^{-1}\partial_v\partial_u(r\phi) + mr^{-4}(r\phi) , \nonumber \\
        \textrm{so that} \qquad \partial_u\partial_v(r\phi) &= -\frac{2mf}{r^3}(r\phi) .
    \end{align}
\end{example}

\begin{example}[The $\ell = 1$ mode in Schwarzschild]
    A less trivial example is given by the $\ell = 1$ case. Let $\phi = \phi(u,v)Y_{\ell=1}$ be a solution to the scalar wave equation on Schwarzschild supported on the $\ell = 1$ mode. Then
    \begin{align}
        0 &= \th_c\th'^2_c\phi - \th'_c\Psi_2\phi - 7\Psi_2\th'_c\phi \nonumber \\
        &= \tfrac{1}{2}r^{-1}f^{-2}\partial_v(r^{-2}f\partial_u(r^2f^{-1}\partial_u(r\phi))) + mr^{-5}(r\phi) + 7mr^{-4}f^{-1}\partial_u(r\phi) , 
    \end{align}
    so that
    \begin{align}
        \partial_v(r^{-2}f\partial_u(r^2f^{-1}\partial_u(r\phi))) &= -2mr^{-4}f^2(r\phi) - 14mr^{-3}f\partial_u(r\phi) . \label{eq:l1_conservation_law}
    \end{align}
    As an extra check, some more algebra easily reveals that this is consistent with Kehrberger's approximate conservation law (cf. Eq. (3.15) of \cite{Kehrberger2021c})
    \begin{align}
        \partial_v(r^{-2}\partial_u(r^2\partial_u(r\phi) + mr\phi)) &= 6m^2r^{-5}f(r\phi) - 12mr^{-3}f\partial_u(r\phi) . \label{eq:Kehrberger_conservation_law}
    \end{align}
\end{example}

\begin{example}[The $\ell = 2$ mode of the Teukolsky equation]
    Let $\psi_0 = \psi(u,v){}_2Y_{\ell=2}$ be a solution to the Teukolsky equation in Schwarzschild supported on the $\ell = 2$ mode. Then
    \begin{align}
        \th'_c\th_c\psi - 15\Psi_2\psi &= \tfrac{1}{2}r^{-1}f^{-1}\partial_u(r^{-4}f^2\partial_v(r^5f^{-2}\psi)) + 15mr^{-4}(r\psi) ,
    \end{align}
    so that
    \begin{align}
        \partial_u(r^{-4}f^2\partial_v(r^5f^{-2}\psi)) &= -30mr^{-2}f\psi .
    \end{align}
    Note that, ignoring derivatives, the left-hand side contains $r^{-4}\times r^5 = r^1$, while the right-hand side contains $r^{-2}$. Hence, the left-hand side contains three extra powers of $r$, as required (cf. Eq. (5.5) of \cite{Kehrberger2024a}). 
\end{example}

\section{Acknowledgements}

I would like to thank Neev Khera for stimulating discussions. I would also like to thank Béatrice Bonga and Eric Poisson for reading the original manuscript and providing helpful comments.  This work was supported by the Natural Sciences and Engineering Research Council of Canada.

\bibliography{references}

\appendix

\section{The GHP formalism}
\label{appendix:GHP}

The basis of the GHP formalism is the complex null tetrad $(l^a, n^a, m^a, \bar{m}^a)$ satisfying $l_an^a = 1$ and $m_a\bar{m}^a = -1$. The connection coefficients are the $12$ complex functions defined as
\begin{align}
    \boxed{\begin{matrix}
        \kappa & \epsilon & \tau' \\
        \rho & \beta' & \sigma' \\
        \sigma & \beta & \rho' \\
        \tau & \epsilon' & \kappa'
    \end{matrix}} \qquad := \qquad
    \boxed{\begin{matrix}
        m^aDl_a & \tfrac{1}{2}(n^aDl_a + m^aD\bar{m}_a) & \bar{m}^aDn_a \\
        m^a\delta l_a & \tfrac{1}{2}(n^a\delta l_a + m^a\delta\bar{m}_a) & \bar{m}^a\delta n_a \\
        m^a\delta'l_a & \tfrac{1}{2}(n^a\delta'l_a + m^a\delta'\bar{m}_a) & \bar{m}^a\delta'n_a \\
        m^aD'l_a & \tfrac{1}{2}(n^aD'l_a + m^aD'\bar{m}_a) & \bar{m}^aD'n_a
    \end{matrix}}
\end{align}
where
\begin{subequations}
\label{eq:tetrad_derivatives}
\begin{align}
    D &:= l^a\nabla_a \\
    \delta &:= m^a\nabla_a \\
    \delta' &:= \bar{m}^a\nabla_a \\
    D' &:= n^a\nabla_a . 
\end{align}
\end{subequations}
The \textit{priming operation} $'$ interchanges $l^a \leftrightarrow n^a$ and $m^a \leftrightarrow \bar{m}^a$. The Lorentz transformations keeping the null directions $l^a$ and $n^a$ fixed are given by
\begin{align}
    (\hat{l}^a, \hat{n}^a, \hat{m}^a, \hat{\bar{m}}^a) = (\lambda\bar{\lambda}l^a, \lambda^{-1}\bar{\lambda}^{-1}n^a, \lambda\bar{\lambda}^{-1}m^a, \lambda^{-1}\bar{\lambda}\bar{m}^a) .
\end{align}
A function $\eta$ is said to be a GHP-weighted scalar with weights $\{p, q\}$ if, under such a transformation, $\hat{\eta} = \lambda^p\bar{\lambda}^q\eta$. The weighted connection coefficients, and their respective weights, are given by
\begin{subequations}
\begin{align}
    &\kappa : \{3,1\} \\
    &\sigma : \{3,-1\} \\
    &\rho : \{1,1\} \\
    &\tau : \{1,-1\} , 
\end{align}
\end{subequations}
and their primed variants. Priming and complex conjugating changes the weights according to $\{p,q\}' = \{-p,-q\}$ and $\overline{\{p,q\}} = \{q,p\}$. The remaining coefficients combine with the derivatives \eqref{eq:tetrad_derivatives} to form the weighted operators
\begin{subequations}
\label{eq:GHP_operators}
\begin{align}
    \th &:= D - p\epsilon - q\bar{\epsilon} : \{1,1\} \\
    \eth &:= \delta - p\beta + q\bar{\beta}' : \{1,-1\} \\
    \eth' &:= \delta' + p\beta' - q\bar{\beta} : \{-1,1\} \\
    \th' &:= D' + p\epsilon + q\bar{\epsilon} : \{-1,-1\} .
\end{align}
\end{subequations}
The curvature scalars are defined as
\begin{align}
    \boxed{\begin{matrix}
        \Phi_{00} & \Phi_{01} & \Phi_{02} \\
        \Phi_{10} & \Phi_{11} & \Phi_{12} \\
        \Phi_{20} & \Phi_{21} & \Phi_{22} 
    \end{matrix}} \qquad := \qquad
    \boxed{\begin{matrix}
        \tfrac{1}{2}R_{ab}l^al^b & \tfrac{1}{2}R_{ab}l^am^b & \tfrac{1}{2}R_{ab}m^am^b \\
        \tfrac{1}{2}R_{ab}l^a\bar{m}^b & \tfrac{1}{2}R_{ab}l^an^b - 3\Lambda & \tfrac{1}{2}R_{ab}m^an^b \\
        \tfrac{1}{2}R_{ab}\bar{m}^a\bar{m}^b & \tfrac{1}{2}R_{ab}\bar{m}^an^b & \tfrac{1}{2}R_{ab}n^an^b
    \end{matrix}}
\end{align}
where $\Lambda := \tfrac{1}{24}R$, and
\begin{align}
\begin{split}
    \Psi_0 &:= C_{abcd}l^am^bl^cm^d \\
    \Psi_1 &:= C_{abcd}l^am^bl^cn^d \\
    \Psi_2 &:= C_{abcd}l^am^b\bar{m}^cn^d \\
    \Psi_3 &:= C_{abcd}l^an^b\bar{m}^cn^d \\
    \Psi_4 &:= C_{abcd}\bar{m}^an^b\bar{m}^cn^d .
\end{split}
\end{align}
The curvature scalars have weights
\begin{align}
    \Lambda &: \{0, 0\} & \Phi_{ij} &: \{2-2i, 2-2j\} & \Psi_k &: \{4-2k, 0\} .
\end{align}
The electromagnetic scalars are defined as
\begin{align}
\begin{split}
    \phi_0 &:= F_{ab}l^am^b \\
    \phi_1 &:= \tfrac{1}{2}F_{ab}(l^an^b + m^a\bar{m}^b) \\
    \phi_2 &:= F_{ab}\bar{m}^an^b .
\end{split}
\end{align}
We may find expressions of the curvature scalars in terms of the connection coefficients by using the definition of the Riemann tensor
\begin{multline}
    R_{abcd}W^aX^bY^cZ^d = \\
    Y^cW^a\nabla_a(X^b\nabla_bZ_c) - Y^cX^a\nabla_a(W^b\nabla_bZ_c) - Y^c(X^b\nabla_bW^a - W^b\nabla_bX^a)\nabla_aZ_c , 
    \label{eq:Riem}
\end{multline}
and substituting the tetrad vectors for $W^a, X^a, Y^a, Z^a$. The weighted components of (\ref{eq:Riem}) are
\begin{subequations}
\begin{align}
    \th\rho - \eth'\kappa &= \rho^2 + \sigma\bar{\sigma} - \bar{\kappa}\tau - \tau'\kappa + \Phi_{00} \\
    \th\sigma - \eth\kappa &= (\rho + \bar{\rho})\sigma - (\tau + \bar{\tau}')\kappa + \Psi_0 \\
    \th\tau - \th'\kappa &= (\tau - \bar{\tau}')\rho + (\bar{\tau} - \tau')\sigma + \Psi_1 + \Phi_{01} \\
    \eth\rho - \eth'\sigma &= (\rho - \bar{\rho})\tau + (\rho' - \bar{\rho}')\kappa - \Psi_1 + \Phi_{01} \\
    \eth\tau - \th'\sigma &= -\rho'\sigma - \bar{\sigma}'\rho + \tau^2 + \kappa\bar{\kappa}' + \Phi_{02} \\
    \th'\rho - \eth'\tau &= \rho\bar{\rho}' + \sigma\sigma' - \tau\bar{\tau} - \kappa\kappa' - \Psi_2 - 2\Lambda . \label{eq:spin_coeff_f}
\end{align}
\end{subequations}
The remaining curvature scalars appear in the commutators of the GHP operators \eqref{eq:GHP_operators}
\begin{subequations}
\begin{align}
    [\th, \th'] &= (\bar{\tau} - \tau')\eth + (\tau - \bar{\tau}')\eth' + pK^* + q\bar{K}^* \\
    [\th, \eth] &= \bar{\rho}\eth + \sigma\eth' - \bar{\tau}'\th - \kappa\th' + p(\rho'\kappa - \tau'\sigma + \Psi_1) - q(\bar{\sigma}\bar{\kappa} - \bar{\rho}\bar{\tau}' + \Phi_{01}) \\
    [\eth, \eth'] &= (\bar{\rho}' - \rho')\th + (\rho - \bar{\rho})\th' - pK + q\bar{K} , \label{eq:eth_commutator}
\end{align}
\end{subequations}
where
\begin{subequations}
\begin{align}
    K &:= -\rho\rho' + \sigma\sigma' - \Psi_2 + \Phi_{11} + \Lambda \\
    K^* &:= -\kappa\kappa' + \tau\tau' - \Psi_2 - \Phi_{11} + \Lambda .
\end{align}
\end{subequations}
If $l^a$ and $n^a$ are orthogonal to a space-like two-surface $\mathcal{S}$, $K + \bar{K}$ is the Gaussian curvature of $\mathcal{S}$.

\end{document}